\documentclass{llncs}

\usepackage{frameit,amsmath,color,latexsym,graphics,wrapfig,amssymb,flushend,balance}

\usepackage{times,epsfig,amsthm}
\usepackage{graphicx}
\usepackage[title]{appendix}
\usepackage{ifthen}
\usepackage{times,epsfig,amsmath,clrscode3e}
\usepackage{frameit,color,latexsym,graphics,wrapfig,textcomp,dsfont,centernot,enumitem,balance,multirow}
\usepackage[utf8]{inputenc}

\usepackage{hyperref}
\usepackage{cleveref}
\crefname{section}{§}{§§}
\Crefname{section}{§}{§§}
\usepackage{flushend}
\usepackage{url}
\usepackage {tikz}
\usetikzlibrary{arrows,automata}
\usetikzlibrary{positioning}
\usepackage{pifont}
\usepackage{epstopdf}
\usepackage{xcolor,colortbl}
\definecolor{Gray}{gray}{0.85}
\newcommand{\Conv}{\mathop{\scalebox{1.5}{\raisebox{-0.2ex}{$\ast$}}}}
\usepackage[linesnumbered,ruled,vlined]{algorithm2e}

\newcommand{\setvars}[1]{\ensuremath{\bar{#1}}}

\newcommand{\savespace}{\vspace{-2mm}}
\newcommand{\saveone}{\vspace{-1mm}}

\newcommand{\RA}{\ensuremath{\mathcal{R}}}

\newcommand{\inv}{\myit{invert}}

\newcommand{\error}{\ensuremath{\varepsilon}}

\mathchardef\mhyphen="2D

\newcommand{\sepnodeF}[3]{\ensuremath{{#1}{\pto}#2(#3)}}

\newcommand{\sepnodeW}[4]{\ensuremath{\code{ST}({#1}{.}#2{,}#3{,}#4)}}
\newcommand{\sepnodeR}[4]{\ensuremath{\code{LD}({#1}{.}#2{,}#3{,}#4)}}
\newcommand{\sepnodeD}[2]{\ensuremath{\code{DEL}({#1}{,}#2)}}
\newcommand{\sepnode}[3]{\ensuremath{#1{\mapsto}#2(#3)}}
\newcommand{\seppredF}[2]{\ensuremath{#1(#2)}}
\newcommand{\seppred}[2]{\ensuremath{#1(#2)}}

\newcommand{\base}{\ensuremath{{B}}}

\newcommand{\report}[1]{ }

\newcommand{\acm}[1]{ }

\newcommand{\hide}[1]{}
\newcommand{\hideie}[1]{}
\newcommand{\der}{\ensuremath{\texttt{-\!>}}}
\newcommand{\nil}{{\code{null}}}
\newcommand{\res}{\btt{res}}
\newcommand{\emp}{\btt{emp}}

\newcommand{\pure}{\ensuremath{\pi}}

\newcommand{\heap}{\ensuremath{\kappa}}

\newcommand{\constr}{\ensuremath{\Phi}}

\newcommand{\myit}[1]{\textit{#1}}

\def\sep{\code{*}}

\def\primeV{\myit{fresh}}

\newcommand{\rulen}[1]{\ensuremath{{\bf \scriptstyle [\underline{#1}]}}}

\newcommand{\dangl}[1]{\ensuremath{#1{\not\hookrightarrow}\_}}

\def\FV{\myit{FV}}

\def\D{\Delta}

\def\new{\btt{new}}

\def\bool{\code{bool}}
\def\int{\code{int}}

\def\true{\code{true}\,}
\def\false{\code{false}\,}
\def\a{a}

\newcommand{\myif}[3]{\code{if}~#1~\code{then}~#2~\code{else}~#3}

\newcommand{\code}[1]{{\small {\ensuremath{\tt #1}}}}
\newcommand{\sm}[1]{{\small \mbox{$#1$}}}
\newcommand{\btt}[1]{{\ensuremath{\tt #1}}}

\newcommand{\locnay}[1]{}
\newcommand{\wnnay}[1]{}
\newcommand{\scnay}[1]{}
\newcommand{\nodo}[1]{}

\newcommand{\defpred}{\btt{pred}}

\numberwithin{thm}{section}

\newtheorem{defn}{Definition}

\def\S{\Sigma}

\def\nochange{\myit{initV}}

\def\sep{\ensuremath{*}}
\newcommand{\pto}{{\scriptsize\ensuremath{\mapsto}}}

\newcommand{\hlr}[3]{\ensuremath{\rulen{#1}\frac{\begin{array}{c}
#2\end{array}}{#3}}}
\renewcommand{\hlr}[3]{\ensuremath{\frac{\begin{array}{c}
\rulen{#1}\\
#2\end{array}}{#3}}}
\newcommand{\hlrside}[4]{\ensuremath{\rulen{#1}\frac{\begin{array}{c}
#2\end{array}}{#3}{#4}}}
\renewcommand{\hlrside}[4]{\ensuremath{\frac{\begin{array}{c}
\rulen{#1}\\#2\end{array}}{#3}{#4}}}
\newcommand{\hlrnone}[3]{\ensuremath{\frac{\begin{array}{c}
#2\end{array}}{#3}}}

\newcommand{\htriple}[3]{\ensuremath{ \code{exec}(#1{,}#2) {\yields} #3}}

\def\Flds{\myit{Fields}}
\def\Store{\myit{Heaps}}

\def\Stack{\myit{Stacks}}
\def\Locations{\myit{Loc}}
\def\Val{\myit{Val}}
\def\Var{\myit{Var}}

\def\iffs{\small \btt{ iff~}}

\def\inv{\bigtriangledown}

\newcommand{\atom}{\alpha}

\newcommand{\yields}{\leadsto}

\newcommand{\anon}{\ensuremath{\_\,}}

\newcommand{\imply}{\ensuremath{\Rightarrow}}

\def\Dns{\myit{Node}}
\newcommand{\defsym}{\ensuremath{\overset{\text{\scriptsize{def}}}{=}}}
\def\emptyheap{\sc{e}}
\newcommand{\sheaps}{\ensuremath{h}}
\newcommand{\sstack}{\ensuremath{s}}
\newcommand{\force}{\ensuremath{\models}}

\newcommand{\invo}{\ensuremath{\sc{ \overline{inv}}}}

\newcommand{\form}[1]{\ensuremath{#1}}

\newcommand{\horn}{\ensuremath{\sigma}}

\newcommand{\PName}{\mathcal{P}}
\newcommand{\ass}{\ensuremath{\imply}}

\newcommand{\basetheory}{\ensuremath{\mathcal{L}}}

\newcommand{\slap}{\ensuremath{\code{SLPA_{ind}+}}}

\newcommand{\toolname}{S2TD}
 \newcommand{\sat}{{\sc{\code{yes}}}}
 \newcommand{\unsat}{{\sc{\code{no}}}}

\newcommand{\algonamesl}{{\sc{\code{{S2SAT_{SL}}}}}}
\newcommand{\algonameesl}{{\sc{\code{{S2e_{SL}}}}}}

\newcommand{\lbsl}{\code{link\_back_{eSL}}}
\newcommand{\uosl}{\code{UA_{eSL}}}
\newcommand{\oosl}{\code{OA_{eSL}}}
\newcommand{\unfoldsl}{\code{unfold_{eSL}}}

\def\qed{\hfill\ensuremath{\square}}

\newcommand{\satb}[1]{\code{\bf sat}(\ensuremath{#1})}
\newcommand{\xpure}{{\bf {\scriptstyle eXPure}}}

\newcommand{\sstate}{\Gamma}

\newcommand{\trace}{\tau}

\def\structd{\ensuremath{struct~node}}

\newcommand{\utree}[1]{\ensuremath{{\mathcal T}_{#1}}}

\newcommand{\backfun}{\ensuremath{{\mathcal C}}}
\newcommand{\ctree}[3]{\ensuremath{{\mathcal C}(#1{\rightarrow}#2, #3)}}
\newcommand{\sub}{\ensuremath{\theta}}

\newcommand{\conferencepaper}{1} 
\newcommand{\rep}[1]{\ifthenelse{\conferencepaper = 0}{#1}{}}
\newcommand{\repconf}[2]{\ifthenelse{\conferencepaper = 0}{#1}{#2}}

\begin{document}


 \title{{\toolname}: a Separation Logic Verifier that Supports Reasoning of the Absence and Presence of Bugs}


 \author{Quang Loc Le$^1$ \and Jun Sun$^2$ \and Long H. Pham$^2$ \and Shengchao Qin$^3$}
 \institute{University College London, United Kingdom \and Singapore Management University, Singapore \and Huawei Hong Kong Research Center, Hong Kong, China}

%
%
%
%


\maketitle

\begin{abstract}
Heap-manipulating programs
are 
known to be challenging
to reason about.
We present a novel verifier for heap-manipulating programs called {\toolname}, which 
encodes  programs systematically in the form of Constrained Horn Clauses (CHC) 
using a novel
extension of separation logic (SL) with recursive predicates
and dangling predicates.
  {\toolname} actively explores cyclic proofs 
to address the path-explosion problem.
{\toolname} differentiates itself from existing CHC-based verifiers by focusing on
heap-manipulating programs and by employing cyclic proof to efficiently verify
or falsify them with counterexamples.
Compared with  existing SL-based verifiers,
{\toolname}
precisely specifies the heaps over
 de-allocated pointers
 to avoid false positives in reasoning about the presence of bugs.
{\toolname} has been evaluated using a comprehensive set of benchmark programs from the SV-COMP repository. The results show that
{\toolname} is more effective than
state-of-art program verifiers
and is
more efficient than most of them.

\hide{Symbolic execution is a well established program analysis method (e.g., for test input generation or program verification).
However, one of its
primary limitations is path-explosion that degrades the generation of high coverage inputs.
Although there have been several ideas proposed to attack this problem for non-heap domains,
it lacks a proposal for heap-manipulating programs.

In this work, we propose a novel approach to solving the path-explosion in symbolic execution of
heap-manipulating programs with cyclic proofs. Particularly,
cyclic proofs help to detect subsuming paths that can be guaranteed to not hit a bug.
We demonstrate this proposal in verifying heap-manipulating programs using Constrained Horn Clause
in which separation logic is used to capture heap abstraction.
We show the challenges and our solutions to overcome.
We have implemented our proposed verification system and applied it
to a set of heap-based benchmarks. The experimental results show that our proposal
is effective and efficient compared with
state-of-the-art verifiers.}


\hide{Constrained Horn Clause (CHC)
 has recently
emerged as a promising intermediate language for the verification of unbounded
programs.
 Its main benefit is the separation of concerns, i.e., the same
system can be extended to support different programming
language syntaxes, operational semantics or verification algorithms. 
Unfortunately,
existing CHC-based systems are designed for non-heap logic and provide little support for dynamically linked data structures, which are essential for many real-world programs.}
\hide{a separation logic based CHC system for verifying or falsifying
safety properties of
heap-manipulating programs.
Given a sequential C program annotated with assertions (and without any invariant annotations), our system automatically
generates conditions for the verification of memory safety
(i.e., no null dereference or use-after-free or double-free) 
and correctness of the assertions.
These conditions are systematically
encoded in the proposed CHCs 
which are then discharged by a novel decision procedure for both
error detection and verification (with cyclic proofs).}

\end{abstract}


\section{Introduction}
Heap-manipulating programs are often building-blocks of real-world applications.
Furthermore,
heap-related bugs are the sources of many security vulnerabilities.
The correctness of heap-manipulating
 programs is thus of great importance and yet they are notoriously challenging to verify 
 \cite{Cadar:CACM:2013}. 
Existing verification techniques  focus only on either numerical programs
or  safety verification
of heap-manipulating programs.
Moreover, existing techniques reason about the presence of memory bugs
either by  supporting only (in)equality constraints for pointers  \cite{Sen:FSE:05}
or by modeling pointers through the theory of arrays \cite{Cadar:OSDI:08,Godefroid:2012}.
None of them could  model
heap-manipulation precisely so as to
effectively and efficiently reason about the presence of heap-related bugs.
To bridge this gap,  we present {\toolname}, a novel verifier for heap-manipulating programs. {\toolname} differentiates itself from existing program verifiers not only by introducing a novel memory model for pointers focusing on
 heap-manipulating programs but also by adopting a novel approach for verification and falsification of heap-manipulating programs.

Firstly, {\toolname} is based on the promising framework of Constrained Horn Clause (CHC). 
Over the years, there has been an increasing interest in logic-based systems using CHC~\cite{Hoder:CAV:2011,Nikolaj:SMT:2012,DBLP:conf/birthday/BjornerGMR15,export:180055,Grebenshchikov:PLDI:2012,Gurfinkel:CAV:2015},
that provide a genuine compound of program verification techniques.
Given a program,
a logic-based system partitions the
verification process into two phases: (1) generating verification
 conditions in the form of CHC from the program and (2) checking the properties (e.g., reachability) of the CHC using
 decision procedures.
As a result, logic-based systems can take advantage of rapidly developing satisfiability checking engines,
 especially SMT solvers \cite{Barrett:CAV:2011,TACAS08:Moura} as the decision procedure.


However, existing logic-based systems provide little support for reasoning about heap-manipulating programs.
Fundamentally, this is because they rely on SMT solvers which are primarily based on first-order logic and have not catered to the needs of verifying heap-manipulating programs.
Another particular challenge to 
them
is that the encoding of the manipulations over the heaps
must be {\em flow-sensitive}.
For example,
 the CHC encoded for the two-command consequence \code{x{=}malloc(...); i{=}x{\der}val} must preserve the
 order of these two memory actions: first allocating a heap and then reading its content.
Otherwise, a false alarm might be raised. 
To tackle these challenges, {\toolname} equips CHC with a novel
extension of separation logic with recursive predicates
and dangling predicates to support precise reasoning about heap-manipulation.

\hide{Nowadays,
Abstract Interpretation (AI) \cite{Cousot:POPL:1977},
 Model Checking (MC) \cite{Clarke:1981:DSS} and Symbolic Execution (SE) \cite{King:1976:SEP}
are the most predominant automated techniques for software verification.
Each technique has its own advantages. For example, AI is efficient for undecidable properties;
MC is dedicated for complete domains; and SE is a balanced framework for
the purposes of concolic testing and deductive verification.
It is thus tempting to build a verification system which combines these techniques.
This task is challenging 
 since each technique
 relies on its own intermediate representations
and interpretation of program semantics.
An alternative option would be to run a series of tools independently.
 However, this limits the information that can be exchanged
among techniques, which could be used to enhance the capability of the individual techniques.}

Secondly, {\toolname} novelly employs cyclic proof techniques to discharge CHC. One way to verify a heap-manipulating program
is to verify the program paths one-by-one, which is also known as symbolic execution~\cite{King:1976:SEP}.
\hide{Symbolic execution  \cite{King:1976:SEP} is a well-established method in automated software testing.
A symbolic engine executes a program with symbolic values (rather than concrete ones)
and constructs symbolic constraints, called path conditions,
which include the conditions at the branch commands (i.e., \code{if}, \code{while})
along the path executed.
To achieve the high coverage, the engine generate new paths from the existing ones.
A new path is generated by negating a condition of an existing path.
The input for this new path is generated typically using a constraint solver.
A path is feasible if its path condition is satisfiable.
Otherwise, it is infeasible.
A program is buggy if
there is a feasible path that reaches a program buggy point.}
Such an approach is ineffective due to the so-called path-explosion problem.
Several methods have been proposed to attack this problem, e.g., using function summaries \cite{Godefroid:PLDI:2007}, combining static and dynamic analysis~\cite{Boonstoppel:TACAS:2008,Godefroid:POPL:2010}, and using interpolation~\cite{Jaffar:FSE:2013}.
These methods are however not designed with heap-manipulation in mind.
S2TD tackles this problem with cyclic proofs. 
%
%
The idea is to actively explore subsumption relationships between path conditions of different program paths in order to prune
them.
While such circular reasoning has been applied
in safety/termination verification \cite{Brotherston:POPL:08,Brotherston:CADE:17,Brotherston:CPP:17},
{\em this work 
applies cyclic satisfiability proofs into symbolic execution to reason about the presence (and absence) of bugs.}

\hide{First,
assume that the program is annotated with certain bug condition of the form
``\code{if (\form{C}) ~then~\code{ERROR()}}'', where only when
the condition \code{C} evaluated to \form{\true}
along a path, the path is buggy. Secondly, assume that at a program location point \form{l_0},
there are three branches as follows:
where the paths along location points \form{l_1} and \form{l_2} have been symbolically executed
and the remaining dotted branches have not been executed yet.
We further assume that the path along \form{l_1} is annotated with the bug condition \form{C}
and its path condition \form{\delta_1} is unsatisfiable;
 the path along \form{l_2} is not annotated with any bug condition
and its path condition is \form{\delta_2}.
Obviously, these two paths are bug-free.
Then a path is bug-free as well if its path condition \form{\delta_i} is subsumed by either
\form{\delta_1} or \form{\delta_2}. Hence, that path could be pruned.
Furthermore, if (i) all the remaining paths starts at \form{l_0} (i.e., the command at \form{l_o} is either a loop or a call of a recursive function),
(ii) we can represent all paths at \form{l_0} in the form of a constraint formula \form{\D_0}
as well as all the remaining paths in the form of another constraint formula \form{\D_i},
and (iii) \form{\D_i} is subsumed by \form{\D_0},
then we can soundly prune all the remaining paths as well.
In the following, we show how to use Constrained Horn Clause (CHC)
to capture the remaining paths in symbolic execution.}

\hide{There has been considerable progress on automated verification of heap-manipulating programs over the past decades.
Broadly speaking, existing approaches can be classified into two categories.
The first category of verification tools, based on
under-approximation, aims for bug identification, i.e., they search for a witness which shows that a program is erroneous.
Example approaches in this category include  program testing~\cite{Herbert:SEFM:2015},  symbolic execution~\cite{Braione:FSE:2015}
and  bounded model checking~\cite{Wei:VMCAI:2014,Chalupa:TACAS:2016}.
%
%
The second category of verification tools, based on
over-approximation, aims for program verification, i.e., they search for a proof
of unsatisfiability to show that a program is error-free.
Techniques used in these tools include data flow analysis~\cite{Sui:ISSTA:2012}, abstract interpretation~\cite{CAV08:Yang,Dudka:CAV:2011,Bouajjani:PLDI:2011}, and deductive verification~\cite{Cristiano:POPL:09,Chin:SCP:12,Loc:CAV:2014}.
Recently, there have been some efforts to support both program verification and bug identification for heap-manipulating programs.
For examples, the work on spatial interpolant~\cite{Aws:ESOP:2015}, which assumes that a shape
abstraction is given,  constructs an over-approximation and an under-approximation of the abstraction at the same time.
In addition, there are techniques based on hardwired abstraction predicates which work only for specific data structures (e.g., list~\cite{Itzhaky:CAV:2014,Aleksandr:CAV:2015,Ruzica:TACAS:2014,Navarro:APLAS:2013}).
These techniques are not designed for arbitrary data structures.}

{\toolname} proceeds in two phases. In the first phase, it
precisely encodes a program into an intermediate representation. 
The intermediate language is an extension of the existing CHC to support
separation logic~\cite{Ishtiaq:POPL01,Reynolds:LICS02}.
We further extend the standard separation logic with additional
assertions (e.g., dangling predicates).
Without the dangling predicates,
 a reasoning system (e.g., \cite{Berdine:APLAS05,jacm.Calcagno11,Chin:SCP:12,Piskac:CAV:2013,Muller:CAV:2016,Long:ATVA:2019,Long:FM:2019})
 is unable to represent
the de-allocated heaps precisely. Such a system typically
over-approximates the semantics of the heaps and
produces false positives in general.  {\em To the best of our knowledge, {\toolname} is the first 
reasoning system
supporting such extended separation logic.}

In the second phase, the CHC generated in the first phase is solved using
a decision procedure, called  {\algonameesl}.
{\algonameesl} is an instantiation of the general cyclic proof framework presented 
in \cite{Loc:CAV:2016} for
the proposed separation logic.
 {\algonameesl} can be regarded as a symbolic execution engine for our CHC-based system. It executes the encoded program symbolically and constructs an execution tree. It simultaneously maintains 
both under- and over-approximation of the CHC and actively explores cyclic proof as described above.
It is sound for both verification and falsification. A cyclic proof is constructed in the former case and a counterexample 
is generated in the latter case.

\paragraph*{Contribution} We make the following contributions: (i)
We propose a novel verification system to verify or falsify heap-manipulating programs based on CHC. 
(ii) We present an effective
 solver, an instantiation of
the general framework in \cite{Loc:CAV:2016},  
for the CHC based on cyclic proofs.  
(iii)  We have implemented {\toolname} 
 and have applied it to a comprehensive set of heap-manipulating programs.
The experimental results show that {\toolname} is effective and efficient compared with
state-of-the-art verifiers.



 \section{Intermediate Language} \label{sec:assert}
We first introduce the intermediate language 
we use to  encode programs as CHC. As shown in Fig. \ref{prm.spec.fig}, 
it is based on
separation logic 
with
{\em inductive} predicates and arithmetic.

The intermediate language  formulas and definitions of inductive predicates. A formula \form{\constr} is a disjunction of symbolic heap \form{\D} where each disjunct models one program 

\hide{
given that there is
    an infinite collection of variables $\Var$.  \form{\bar{x}} denotes a sequence of variables
and \form{{x_i}} its \form{i^{th}} element,
 a finite collection of data structures {\Dns},
 a finite collection of fields {\Flds},
 a disjoint set {\Locations} of locations (i.e., heap addresses),
 a set of non-address values {\Val} such that \form{\mathbb{Z} {\subset} \Val}, \form{\nil{\in} {\Val}}
and {\Val} \form{\cap} {\Locations} {=}\form{\emptyset},
 and a finite set of inductive symbols \code{\PName}.
Lastly, $\bot \in \Val$ is a preserved value which
 denotes the content of a heap after it has been
de-allocated.
}

 \begin{wrapfigure}{l}{0.55\textwidth}
 \begin{small}
 \begin{center}
 \savespace 
 \savespace
 \saveone\[\savespace
\begin{array}{rcl}
\constr & {~::=} & \true  ~|~ \false  \mid \D \mid \constr {\vee} \constr 
\\
\D & {~::=}&  \exists \bar{v}{.}~(\heap{\wedge}\pure) 
 \\
\heap & {~::=} & \emp ~|~ 
\sepnodeF{x}{c}{\overline{f{:}v}} ~|~ \seppredF{\code{P}}{\setvars{v}} 
~|~\heap {\sep}\heap \\
\pure & {~::=} &  \true |~ v |~
\pure {\wedge} \pure |~ \pure {\vee} \pure |~ \neg\pure |~
 \exists v{.}~ \pure \mid  
 \atom \mid  \form{\phi}     \\ 
\atom  & {~::=} &  
v{=}\nil ~|~  \dangl{v} \mid
 \form{\sepnodeR{v}{f}{x}{k}} 
  \\
 & &
  \mid \form{\sepnodeW{v}{f}{x}{k}} 
\mid \form{\sepnodeD{v}{k}} \\
 \form{\phi} & {~::=} & \a {=} \a \mid \a {\leq} \a \quad
a  {~::=}
          ~k \mid v \mid 
          \a {+} \a \mid - \a \\
\textit{CHC} & {~::=} & \D{:}l \imply  \seppredF{\code{P}}{\setvars{v}} \mid \seppredF{\code{P}}{\setvars{v}} \equiv \bigvee^n_{i=1} (\D_i{:}l_i) \mid \D \\
\end{array}
\]
 \saveone
\caption{Syntax of intermediate language (\form{k \in \mathbb{Z}})}\label{prm.spec.fig}
\end{center} 
\savespace
\end{small}
\end{wrapfigure}

path.
\form{\D} is an existentially quantified formula
consisting of a spatial constraint \form {\heap}
 and a pure (non-heap) constraint \form{\pure}.
\form{\FV({\D})} denotes all free variables in the formula \form{\D}.
We assume an infinite collection of variables $\Var$,  a finite collection of data structures {\Dns},
 a finite collection of fields {\Flds},  a  set of heap addresses {\Locations}, 
  a set of non-address values {\Val} i.e. \form{\mathbb{Z} {\subset} \Val}, \form{\nil{\in} {\Val}}
and {\Val} \form{\cap} {\Locations} {=}\form{\emptyset},
 and a finite set of inductive symbols \code{\PName}.
 \form{\bar{x}} denotes a sequence of variables
and $\bot {\in} \Val$ is a preserved value denoting the content of a heap cell following
de-allocation.


A spatial formula {\heap} is either the predicate \form{\emp}
(asserts an empty heap), a points-to predicate \form{\sepnodeF{x}{c}{\overline{f{:}v}}}
 where $c {\in} {\Dns}$ and  $f {\in} {\Flds}$
(asserts that the pointer \form{x}
 points to singly allocated heap typed \form{c} with content \form{\overline{f{:}v}}), an inductive predicate instance \form{\seppredF{\code{P}}{\setvars{v}}} (represents an infinite set of allocated
objects which are defined by predicate \code{P}),
or their spatial conjunction.
When there is no ambiguity, we discard \form{f} and simply write the short form \form{\sepnodeF{x}{c}{\bar{v}}}.
A formula is said to be a {\em base formula}, denoted by \form{\base}, if it does not contain any inductive predicates.
A pure formula {\pure} can be 
a formula in Presburger arithmetic \form{\phi}, a pointer constraint  \form{\atom},  or a boolean combination of them.
The pointer constraint may include (in)equalities, dangling predicate \form{\dangl{x}}
(i.e., \form{x} is not yet allocated or is already de-allocated), or the following three new predicates to simulate
 memory accesses.
\form{\sepnodeR{v}{f}{x}{k}} simulates \form{k^{th}}
memory access for
 loading the
value at the memory of field \form{f}
pointed to by \form{v} into variable \form{x}.
\form{\sepnodeW{v}{f}{x}{k}}  simulates \form{k^{th}}
memory access for writing the value of variable
\form{x} into the memory of field \form{f}
pointed to by \form{v}.
And \form{\sepnodeD{v}{k}} simulates \form{k^{th}}
memory access for de-allocating memory
pointed to by variable \form{v}.
We use \form{\a_1 {\neq} \a_2}
as short forms for \form{\neg (\a_1{=}\a_2)}.
$\res$ is a reserved variable to denote
the returned value of each  procedure.
Note that in our definition, a pure formula 
can be negated. However, our encoding
only generates positive form of the three simulation predicates.




An inductive predicate 
is defined as
$\begin{small}\;\form{
\defpred~ \seppred{\code{P}}{\setvars{t}}~{\equiv}
   \bigvee^n_{i=1} (
   \D_i{:}l_i), \;
}\end{small}$
where
 \code{P} is the predicate name,
\form{\setvars{t}} is a sequence of parameters and
 \form{
 \D_i} (\form{i \in \{1...n\}})
 are
 definition rules (branches). 
%
Each inductive predicate is associated with an  invariant \form{\overline{inv}} representing a superset of all possible models of the predicate.
This invariant is generated automatically
and used to efficiently prune
infeasible executions.
Each branch corresponds to a path in the program captured
by path label \form{l_i},
 a sequence of controls at the {\em branching} statements. 
This label \form{l_i} is used to generate the witness to
 program errors.
 We use 1 to label the branch which satisfies
 conditions of branching statements (e.g., \code{then} branch)
 and use 2, otherwise
(e.g., \code{else} branch).
In each definition rule, variables not in \form{\setvars{t}} are always existentially quantified.
Lastly, a CHC \form{\horn} is defined as:
 \form{\D{:}l {\ass}  \seppredF{\code{P}}{\setvars{v}}} where
\form{\setvars{v}{\subseteq} \FV(\D)}, \form{\D} is its body
 (with path label \form{l}) and
\form{\seppredF{\code{P}}{\setvars{v}}} is its head.
A clause without a head is called a query.
We note that CHCs with the same head e.g., VC generated for different paths of a procedure \code{P} (or a loop)
\form{(\D_{1}{:}l_1 {\ass} P(\setvars{v})) ~{\vee}~
 (\D_{2}{:}l_2 {\ass} P(\setvars{v}))},
is written in the equivalent form of predicate definition
as:
 \form{ P(\setvars{v}) ~{\equiv}~ \D_{1}{:}l_1 ~{\vee}~ \D_{2}{:}l_2}.

\hide{Our proposal relies on the following definitions.
\form{\seppredF{\code{P}}{\setvars{v}}} is called (heap) observable if
there is
at least one free {\em pointer-typed} variable in \setvars{v}.
Otherwise, it is called {\em unobservable}.
\sepnode{v}{c}{\setvars{t}} is called (heap)
 observable if \form{v} is free.
Otherwise, it is {\em unobservable}.}


\paragraph*{Semantics} Formulas of our separation logic fragment are interpreted over pairs
(\form{\sstack},\form{\sheaps}) where \form{\sstack}
models the program stack and \form{\sheaps} models the program heap.
Formally,  we define:
\savespace\[\saveone
\begin{small}
\begin{array}{lcllcllcllcl}
{\Store}  & {\defsym} &  {\Locations} {\rightharpoonup_{fin}} ({\Dns} ~{\rightarrow}~ \Flds ~{\rightarrow}~ \Val \cup \Locations)  &\qquad
{\Stack} & {\defsym} &  {\Var} ~{\rightarrow}~ \Val \cup \Locations
\end{array}
\end{small}\]
The semantics of a formula \form{\constr} is defined by a 
 relation:
\form{\sstack,\sheaps \force \constr}
that forces the stack \form{\sstack \in {\Stack}} and heap \form{\sheaps \in {\Store}} to satisfy the constraint
 \form{\constr}.
The semantics of all predicates except the dangling predicate
and new simulation predicates
are standard (e.g.,~\cite{Loc:CAV:2016} for a reference).
The semantics of the new predicates is as follows.
\[\begin{small}
\begin{array}{lcl}
 {\sstack} {\force} \dangl{v} &
{\iffs}& \forall \sheaps,\code{c}, \code{f} ~s.t. \text{dom}(\sheaps) {=}\{\sstack(v)\},
\sheaps(\sstack(v))(c{,}f) \text{ is defined }, 
 \sheaps(\sstack(v))(c{,}f){=} \bot \\
\form{{\sstack} {\force} \sepnodeR{v}{f}{x}{k}} &
{\iffs}& \form{k \in \mathbb{Z}, k{\geq}0} \text{ and } \form{ \exists \sheaps, \code{c}. ~\form{\sstack},\form{\sheaps} {\force} \sepnodeF{v}{c}{\anon}} \\
\form{{\sstack} {\force} \sepnodeW{v}{f}{x}{k}} &
{\iffs}& \form{k \in \mathbb{Z}, k{\geq}0} \text{ and } \form{ \exists \sheaps, \code{c}. ~\form{\sstack},\form{\sheaps} {\force} \sepnodeF{v}{c}{\anon}} \\
\form{{\sstack} {\force} \sepnodeD{v}{k}} &
{\iffs}& \form{k \in \mathbb{Z}, k{\geq}0} \text{ and } \form{ \exists \sheaps, \code{c}. ~\form{\sstack},\form{\sheaps} {\force} \sepnodeF{v}{c}{\anon}} \\
\end{array}
\end{small} \]

\hide{
%
In our system, pure domains include integer domain (\code{Ints}), 
 and boolean.
The evaluation for pure expressions are determined by valuations as follows:
\saveone\[\saveone
\form{{\sstack}(\a) ~{\in}~ \code{Ints} 
 \qquad {\sstack}(\myit{b}) ~{\in}~ \{\true,\false\}}
\]

The semantics is presented in Fig. \ref{fig.sl.sem} where the semantics for pure formulas
is omitted. 
%
\begin{figure}[t]
\begin{center}
\hide{\[
\begin{array}{rl}
 \form{\sstack} \force \pure_1 {\wedge} \pure_2 ~\iffs &
 \form{\sstack} \force \pure_1 \text{ and }
 \form{\sstack}\force \pure_2 \\
\form{\sstack} \force \pure_1 {\vee} \pure_2 ~ \iffs &
 \form{\sstack} \force \pure_1 \text{ or }
 \form{\sstack}\force \pure_2 \\
 \form{\sstack} \force v_1{\oslash}v_2 ~ \iffs & \force {\sstack}(v_1)~{\oslash}~ {\sstack}(v_2) \text{, where }
 \oslash \in \{=,\neq\}
  \\
 \form{\sstack} \force {\a_1}{\oslash}{\a_2} ~ \iffs & \force {\sstack}(\a_1)~{\oslash}~ {\sstack}(\a_2) \text{, where }
 \oslash \in \{=,\leq\}
  \\
 \end{array}
 \]}
 \[
 \begin{array}{rcl}
\form{\sstack},\form{\sheaps} {\force} \emp 
&{\iffs} & \sheaps {=} \emptyset\\
\!\!\form{\sstack},\form{\sheaps} {\force} \sepnodeF{v}{c}{{f_1{:}v_1..f_n{:}v_n}}
\!\!\!\!&{\iffs} & l {=}\sstack(v), \text{dom}(\sheaps) {=}\{l\} {\wedge}\sheaps(l){=}r \\
&  & \mbox{ and } r(c{,}f_i){=} \sstack({v_i})  \mbox{ for } 1{\le}i{\le}n\\
\form{\sstack},\form{\sheaps} {\force} \dangl{x}
&{\iffs} & \exists \form{\sheaps'}.~  \sheaps' {\#} \sheaps \text{ and }
\form{\sstack},\form{\sheaps'} {\force} \dangl{x} \\
\form{\sstack},\form{\sheaps} {\force}  \heap_1 \sep \heap_2 ~
&{\iffs} & \exists \sheaps_1,\sheaps_2 {.}~ \sheaps_1 {\#} \sheaps_2\mbox{ and }  \sheaps{=}\sheaps_1 {\cdot} \sheaps_2  \\
&&
 \mbox{ and } \sstack,\sheaps_1 \force \heap_1 \mbox{ and } \sstack{,}\sheaps_2 \force \heap_2\\
\form{\sstack},\form{\sheaps} \force \true 
&{\iffs} & \mbox{always}\\
\form{\sstack},\form{\sheaps} \force \false
&{\iffs} & \mbox{never}\\
\form{\sstack},\form{\sheaps} \force \heap{\wedge}\pure
&{\iffs} &  {\sstack} {,}\form{\sheaps} \force \heap
\text{ and } {\sstack} \force {\pure} \\
\form{\sstack},\form{\sheaps} \force \exists v {.}\D
&{\iffs} & \exists {\nu} {.}~ {\sstack}[{v} {\pto} {\nu}] {,}\form{\sheaps} \force \D\\
\form{\sstack},\form{\sheaps} \force  \constr_1 \vee \constr_2
&{\iffs} & \sstack,\sheaps \force \constr_1 \mbox{ or } \sstack,\sheaps \force \constr_2\\
\end{array}
\]
\vspace*{-2mm}
\caption{Semantics. }
\label{fig.sl.sem}
\vspace*{-4mm}
\end{center}
\end{figure}
\form{dom(f)} is the domain of the function  \form{f}; 
 \form{\sheaps_1 {\#} \sheaps_2} denotes that
 heaps $h_1$ and $h_2$ are (domain) disjoint, i.e.,
\form{\text{dom}(\sheaps_1) {\cap} \text{dom}(\sheaps_2) {=} \emptyset}; and
  \form{\sheaps_1 {\cdot} \sheaps_2} denotes the
  union of two disjoint heaps. \sm{{\sstack}[{v_1} {\pto} {\nu_1},..,{v_n} {\pto} {\nu_n}] (v) = {\sstack}(v)} for \sm{v{\notin}\{v_1,..,v_n\}}, and \sm{{\sstack}[{v_1} {\pto} {\nu_1},..,{v_n} {\pto} {\nu_n}] (v_i)=\nu_i} for \sm{1{\le}i{\le}n}.
%
}

\hide{\begin{verbatim}

\eval

\inv^0(lsegn) = \emptyset
\inv^1(lsegn) = \inv^0(lsegn) \vee \eval[(\constr)](\inv^0(lsegn))
               = \emptyset \vee ({(p,p,0,{\emptyheap}) | p \in \val} \vee \emptyset)
                      = {(p,p,0,{\emptyheap}) | p \in \val}
\inv^2(lsegn) = \inv^1(lsegn) \vee \eval[(\constr)]_{\inv^1(lsegn)} = \inv^1 \vee ()
              = {(p,p,0,{\emptyheap}) | p \in \val} \vee
                {(p,p,0,{\emptyheap}) | p \in \val}
                 \vee {(x,p,n,[x->(v,q)]) | x in val, p \in val,
                  n\in val, v,q\in val and x\neq p and n=1 }
\end{verbatim}
}

\hide{
We use $[ \seppredF{\code{P}}{\setvars{v}} ]$ to denote
the set of all models satisfying predicate \form{\seppredF{\code{P}}{\setvars{v}}}
(i.e., \form{\sstack, \sheaps \force \seppredF{\code{P}}{\setvars{v}}} iff
$(\sstack, \sheaps) \in [ \seppredF{\code{P}}{\setvars{v}} ]$).
And \form{\Gamma} is a function to map every inductive predicate \form{\seppredF{\code{P}}{\setvars{v}}}
to $[ \seppredF{\code{P}}{\setvars{v}} ]$.
%
%
Lastly, a CHC \form{\horn} is defined as:
 \form{\D{:}l_i \ass  \seppredF{\code{P}}{\setvars{v}}} where
\form{\setvars{v}{\subseteq} \FV(\D)}, \form{\D} is its body
 (with path label \form{l_i}) and
\form{\seppredF{\code{P}}{\setvars{v}}} is its head.
A clause without a head is called a query.
Semantically, \form{\horn} is satisfiable if for any model \form{\sstack, \sheaps} and
\form{\sstack, \sheaps \force \D}, then \form{(\sstack, \sheaps) \in [\seppredF{\code{P}}{\setvars{v}} ]}. 
A set of CHCs \form{\RA} is satisfiable if there exists a function
\form{\Gamma} 
such that
every clause \form{\horn {\in} \RA} is
satisfiable under \form{\Gamma}. That means any \form{\sstack, \sheaps}, and
\form{{\exists} \setvars{w}.~\Conv^n_{i{=}1}
\sepnodeF{x_i}{c_i}{\setvars{d_i}} {\sep} \Conv^n_{i{=}l}
 \seppredF{\code{P}}{\setvars{v}_i} {\wedge} \pure \imply \seppredF{\code{P_0}}{\setvars{v}_0} {\in} \RA}, \form{\sstack, \sheaps \force {\exists} \setvars{w}.~\Conv^n_{i{=}1}
\sepnodeF{x_i}{c_i}{\setvars{d_i}} {\sep} \Conv^n_{i{=}l}
 \seppredF{\code{P}}{\setvars{v}_i} {\wedge} \pure } and
 \form{\sstack, \sheaps \in [\seppredF{\code{P_0}}{\setvars{v}_0}]} iff there exists
 \form{\sstack_i \subseteq \sstack, \sheaps_i \subseteq \sheaps} for all \form{i \in \{1..l\}},
 \form{\sheaps_i} is disjoint from \form{\sheaps_j} for all \form{i,j \in \{1..l\}}
 and \form{\sstack_i, \sheaps_i \in [\seppredF{\code{P_i}}{\setvars{v}_i}]}.
We note that CHCs with the same head e.g., VC generated for different paths of a procedure \code{P} (or a loop)
\form{(\D_{1}{:}l_i {\ass} P(\setvars{v})) ~{\vee}~
 (\D_{2}{:}l_j {\ass} P(\setvars{v}))},
is written in the equivalent form as:
 \form{ P(\setvars{v}) ~{\equiv}~ \D_{1}{:}l_i ~{\vee}~ \D_{2}{:}l_j}.
\hide{The above semantics indicates that
 a pointer $x$ always satisfies one of the following three predicates: \form{x{=}null}, \form{ \dangl{x}} (i.e., $x$ is dangling), and \form{\sepnodeF{x}{c}{f_i{:}v_i}} (i.e., it points to some heap object).
Based on that we define the negation for pointer constraints 
as follows.
\[
\begin{array}{rcl}
\neg (\sepnodeF{x}{c}{f_i{:}v_i}) & {\equiv} & x{=}\nil ~{\vee}~\dangl{x} \\
\neg (x{=}\nil) & {\equiv} & \sepnodeF{x}{c}{f_i{:}v_i} ~{\vee}~ \dangl{x} \\
\neg (\dangl{x}) & {\equiv} & x{=}\nil ~{\vee}~ \sepnodeF{x}{c}{f_i{:}v_i}
\end{array}
\]}
}

\section{Overall Approach} \label{sec:overall}
\label{sec.motivate}
\hide{\begin{figure*}[t]
 \begin{center}
  \includegraphics[width=0.78\textwidth]{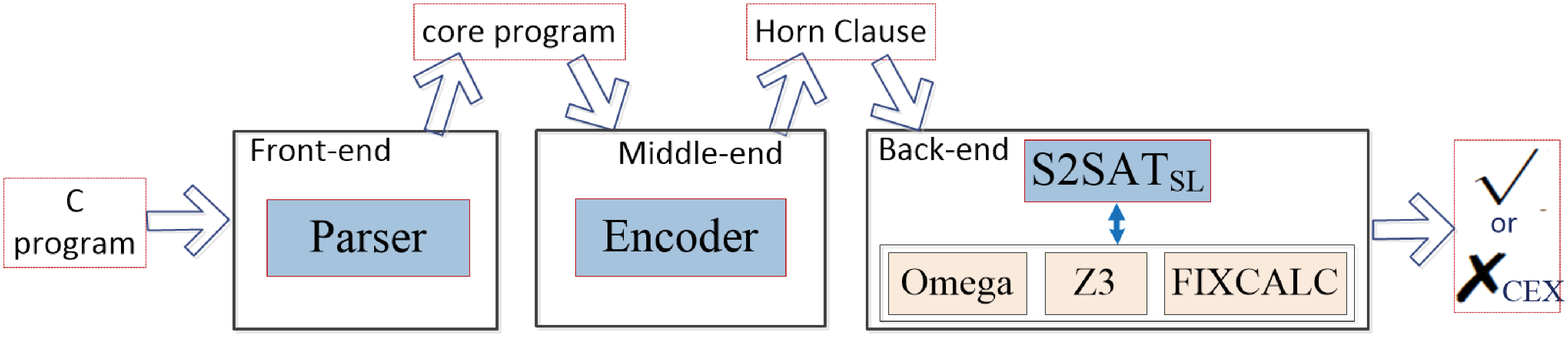}
\savespace  
\end{center}
   \caption{Main Components of {\toolname} Verification System}
   \label{s2sd}
\end{figure*}}

 We show how {\toolname} works using the illustrative C program shown in Fig.~\ref{fig.mov.slldll}\,a). The program allocates a null-terminated singly-linked list using a recursive procedure \code{sll} at line 10. 
The list contains a list of decreasing non-negative integer numbers. The loop from line 
11 to 15 traverses through the list to check if any number stored in the list is negative. For each node in the list, if its value is negative, an error
occurs (at line 13). For simplicity, we assume that user assertions are specified in the form of ``\code{if (\form{C}) ~then~\code{ERROR()}}''.
 Given the program, {\toolname} automatically verifies: (i) whether
method \code{ERROR()} 
at line 
13 is reachable; 
 and (ii) whether this program is memory safe, i.e., no error caused by
 de-reference pointers 
 (at line 5, 13 and 14)
 and free pointers 
 (at line 15). While
{\toolname} could return ``no'' for all these queries,
our experimental results show that the state-of-the-art verification systems (e.g., PredatorHP \cite{Predator:SV:16},
 SeaHorn \cite{Gurfinkel:CAV:2015}, Cascade \cite{Wei:VMCAI:2014},
ESBMC \cite{Cordeiro:2011:ICSE}, UAutomizer \cite{Heizmann:CAV:13},
CPAChecker \cite{Beyer:CAV:2007}, CBMC \cite{Clarke:TACAS:2004} and
Smack-Corral \cite{Haran:TACAS:2015}) could not handle this example.


\hide{
/ list(x)
node prev = x;
while (x != null) {
        x = x-> next;
        x -> prev = prev;
        prev = x;
}
}
\begin{figure*}[tb]
\begin{center} 
\begin{scriptsize}
 \savespace\[\savespace
\begin{array}{ll}
\begin{array}{c}
\begin{array}{ll}
1  & \code{\structd~\{int \, val;} \; \code{ \structd \,{\sep}\, next;\}} \\
2 & \code{\structd{\sep}~sll(int ~ i)\{} \\
3 & \code{~~ if(i{==}0)~ return~ \code{NULL};}  \\
4 & \code{ ~~ else~\{~\structd ~ {\sep}n{=}(\structd{\sep})} \\
  & \code{ \qquad \quad malloc(sizeof(\structd));} \\
5 & \code{ \quad n{\der}val{=}i; ~ n{\der}next{=}sll(i{-}1);} \\
6 & \code{ \quad return~n;~\}~\}} \\
7 & \code{int~main(int ~n)\{} \\
8 & \code{ ~~ \structd~{\sep}x;}  \\
9 & \code{~~ if(n{<}0) ~return~0;}  \\
10 & \code{ ~~ x{=}sll(n);}  \\
11 & \code{ ~~ while(x) \{}  \\
12 & \code{ \quad \structd~{\sep}tmp{=}x;}  \\
13 & \code{ \quad if(x{\der}val{<}0) ~ ERROR();}\\
& \qquad\qquad\qquad
\code{//~assert(x{\der}val{\geq}0);}  \\
14 & \code{ \quad x{=}x{\der}next;}  \\
15 & \code{ \quad free ~ tmp;~\}}  \\
16 & \code{ ~~ return ~1;~\}} \\
\end{array} \\ \\
\text{a) An example C program} \\
\end{array} &
\begin{array}{c}
\begin{array}{l}
\defpred~\form{\seppred{\textit{sll}}{i{,}\res{,}\error}} ~\equiv~
\form{\emp {\wedge}\res{=}\nil {\wedge} i{=}0 {\wedge}\error{=}0} {:}[1] \\ 
~~ {\vee} \form{~{\exists} ~i_1{,}r {.}~ \sepnode{\res}{\textit{node}}{i{,}r} {\sep}
\seppred{\textit{sll}}{i_1{,}r{,}\error}_{1}^0 {\wedge}i{\neq}0{\wedge} i_1{=}i{-}1}{:}[2] 
\\
\invo{:}~i{\geq}0 {\wedge} \error{=}0; \\
\text{} \\
\defpred~\form{\seppred{\textit{loop}_{11}}{x_0{,}x{,}\error}} ~\equiv~
\form{\emp {\wedge} x_0{=}\nil  {\wedge}x{=}x_0 {\wedge} \error{=}0}{:}[2] 
\\
~~ {\vee} ~ \emp{\wedge}{\dangl{x}{\wedge} x{\neq}\nil {\wedge}x{=}x_0 {\wedge} \error{=}1}{:}[1] \\ 
~~ {\vee}~ \exists v_1{.}
\form{\emp{\wedge}{\neg}\dangl{x}
 {\wedge}x{\neq}\nil {\wedge} \sepnodeR{x}{\textit{val}}{v_1}{0}} 
 \\
 \qquad
 \form{{\wedge}x{=}x_0 
 {\wedge}v_1{<}0{\wedge} \error{=}1}{:}[1;1] \\
~~ {\vee}~ \exists v_1, x_1{.}
\form{\seppred{\textit{loop}_{11}}{x_1{,}x{,}\error}_{1}^0{\wedge}{\neg}\dangl{x_0}
 {\wedge}x_0{\neq}\nil} 
 \\
 \qquad
 \form{{\wedge} \sepnodeR{x_0}{val}{v_1}{0} 
 {\wedge}v_1{\geq}0} {\wedge}
\form{ \sepnodeR{x_0}{\textit{next}}{x_1}{1} 
}{:}[1;2] \\
\invo{:}~0{\leq}\error{\leq}1;
\\
\text{} \\
\defpred~\form{\seppred{\textit{main}}{n{,}\res{,}\error}} ~\equiv~
\form{ \emp {\wedge}n{<}0 {\wedge} \res{=}0 {\wedge} \error{=}0}{:}[1] \\
~~ {\vee} ~\form{{\exists} {x_0}{,}x{,}{\error_1}{,}{\error_2} {.}\seppred{\textit{sll}}{n{,}x_0{,}\error_1}_{1}^0
{\sep} \seppred{\textit{loop}_{11}}{x_0{,}x{,}\error_2}_{1}^1 \\
\qquad{\wedge}n{\geq}0
{\wedge} 
}
\form{ \error{=}\error_2 }{:}[2] \\
\invo{:}~0{\leq}\res{\leq}1;\\
\text{ }
\\
\form{\seppred{\textit{main}}{n_0{,}\res_0{,}\error}^0_0 {\wedge}\error{=}1
}. \\
\end{array}\\
\text{ } \\
\text{b) Verification Condition} \\
\end{array}\\
 \end{array}
\]
\end{scriptsize}
\caption{An illustrative example}
\label{fig.mov.slldll}
\end{center}
\end{figure*}

{\toolname} has three main components: a front-end parser,
a middle-end encoder and the back-end decision procedure  {\algonameesl}.
The parser is based on CIL~\cite{Necula:CC:2002}. It takes the C program as input and transforms it
into a core program.
We assume that the transformation starts at a special entry method named \code{main}.
Specifically, 
the parser converts the program into a static single assignment (SSA) form, 
internalizes global variables and transforms loops into
{\em tail recursive} procedures. For example,
 the \code{while} loop from line 11 to 15 is transformed
into a recursive procedure \form{\textit{loop}_{11}}
 and the loop is replaced
by a call of procedure \form{\textit{loop}_{11}}. The arguments of \form{\textit{loop}_{11}} include
a set of input parameters, which are those variables occurring
in the \code{while} condition or the body of the loop,
and a set of output parameters, which are those variables that are
 modified within the loop body (instead of using pass-by-reference variables). 
Each call of \code{ERROR()} is replaced
by: \code{assume (\error{=}1)}
where
\form{\error} is a reserved variable which encodes
  exit conditions.
 \form{\error} has two possible values: \textcolor{red}{\form{\error{=}1}} arising from either \code{ERROR()}
or a heap-manipulation violation
and  \textcolor{blue}{\form{\error{=}0}} for normal termination.
Here,
 we represent the single kind of error.
For a better error explanation, we could consider different kinds of errors using a lattice like in \cite{Loc:NFM:2013}.


\subsection{Flow-Sensitive Encoding}\label{mov.encoder}
The encoder takes a core program as input
and produces a CHC system.
Intuitively, it transforms each procedure into
 an inductive predicate definition.
 It is a forward symbolic executor
of the form \htriple{\D_{pre}}{e}{\D_{post}} (a.k.a a Hoare triple \form{\{\D_{pre}\}e\{\D_{post}\}}), 
 where \form{\D_{pre}} denotes  pre-states, \form{e} is a program command
and \form{\D_{post}} depicts  post-states.
The engine aborts whenever an error 
 is met. 
Particularly, for each procedure \code{mn(\setvars{t})} with a body \form{e} in the program,
the encoder generates the following inductive definition:
 \form{\defpred~mn(\setvars{t},\res,\error)\equiv \constr}
 such that
 \htriple{\emp\wedge \true}{e}{\constr}.
 The arguments of each generated predicate
consist of parameters of the procedure and two additional variables: \form{\res}
and \form{\error}
where
\form{\res} encodes the return value of a procedure
and \form{\error} to capture error status.
In the following, we show how to obtain a flow-sensitive encoding.



We show the
 flow-sensitive encoding
at two levels: inter-procedure and intra-procedure.
%
The former is necessary
as it helps to avoid generating {\em infeasible} executions (e.g., a trace that never calls \code{sll} but executes \code{loop_{11}}).
Particularly,
we annotate every occurrence of inductive predicates,
e.g., \form{\seppred{\textit{main}}{\setvars{v}}}, with
two numbers, e.g., \form{\seppred{\textit{main}}{\setvars{v}}^o_u} where
\form{o} is the order of a callee in the calling context
and \form{u} is the number of unrolling (for
each loop or recursion, whose value is increased by one after being called).
For example, in the main procedure as \code{sll} is always executed
before the loop, the sequence number of \form{\textit{sll}} is 0
and that of \form{\textit{loop}_{11}} is 1.
These numbers ensure that if a pointer is first allocated in \code{sll}
and then accessed in the loop,
there is no memory error.
The unfolding numbers of inductive predicates
 in generated VCs
are initially assigned to 0 if it is not recursive
and 1 otherwise and updated by the solver
during its execution.


For a flow-sensitive intra-procedural
memory access encoding, 
we introduce  a
 local
 order, say \form{k^{th}},
in a sequence of memory accesses in a procedure over
 the three new predicates: 
\form{\sepnodeR{v}{f_i}{x}{k}},
\form{\sepnodeW{v}{f_i}{x}{k}}, and
 \form{\sepnodeD{v}{k}}.
For example, given \form{\D_{pre}} and let \code{e} be a
memory read \sm{x{=}v{\der}f_i},
{S2TD} generates
the following post-state \form{\D_1 {\vee} \D_2 {\vee} \D_3}:
\begin{enumerate}
\item \form{\D_1\equiv \D_{pre} {\wedge} \dangl{v} {\wedge}\textcolor{red}{\error{=}1}}
 encodes a memory
error when  \sm{v} is a dangling pointer. 
\item \form{\D_2\equiv \D_{pre} {\wedge} v{=}\nil{\wedge}\textcolor{red}{\error{=}1}}
 encodes a memory
error when  \sm{v} equals to \form{\nil}. 
\item \form{\D_3\equiv \D_{pre} {\wedge}  v{\neq}\nil {\wedge}
 \neg\dangl{v} {\wedge}\sepnodeR{v}{f_i}{x}{k}} encodes the memory
safety condition when  \sm{v} has been allocated in the pre-state.
\end{enumerate}
The post-states \form{\D_1}, \form{\D_2} and \form{\D_3} are complete and
 pairwise disjoint.
 If \form{\D_{pre}} implies that \form{v} is dangling, 
 \form{\D_1}
is satisfied, but neither \form{\D_2} nor \form{\D_3}.
If \form{\D_{pre}} implies that
 \form{v} has been assigned to \form{\nil}, \form{\D_2}
is satisfied, but neither \form{\D_1} nor \form{\D_3}. Both these scenarios above trigger a memory error.
Otherwise, \form{\D_{pre}} implies that
\form{v} is pointing-to a node, \form{\D_3}
is satisfied, but neither \form{\D_1} nor \form{\D_2} and there is no memory error.
\paragraph{Example Revisited}
The CHC for the program shown in Fig.~\ref{fig.mov.slldll}\,a) is presented in  Fig.~\ref{fig.mov.slldll}\,b).
Each procedure (or loop) is encoded as a predicate definition where
each disjunct corresponds to a program path in the procedure.
To generate a counter example, we keep track of the program paths
using labels. Particularly,
 each disjunct
is attached with a path label (after \form{:}) 1 for the \code{then} (or entering loop) branch and 2 for
the \code{else} (or exiting loop)
 branch. For example, for procedure \code{sll}, the first disjunct of its encoding corresponds to the
 \code{then}-branch
(with the path label [1])
 in Figure~\ref{fig.mov.slldll}, whereas the second disjunct corresponds to the
 \code{else}-branch
(with the path label [2]).

Procedure \code{main} is encoded with the
predicate \form{\textit{main}} whose first disjunct encodes the \code{then}-branch at line 9,
and the second disjunct encodes the else-branch.
The \code{while} loop at lines 11-15 is
encoded as the predicate \form{\textit{loop}_{11}} which
 has four disjuncts.
 The first disjunct encodes the branch
where the loop condition at line 11 does not hold.
The second one encodes the branch
where the loop condition holds and there is a dangling-dereference
error at the access \code{x{\der}val} at line 13.
The third disjunct encodes the if-branch at line 13, i.e., the memory access is safe and the error occurs. 
The last disjunct encodes the else-branch at line 13, i.e., there is no error.
We note that the disjunct corresponding
to null-dereference (i.e., \form{x=\nil}) is infeasible
(and has been discarded)
as it contradicts with the loop condition (i.e., \form{x\neq\nil}).
We further note that the numbers (\form{0} and
\form{1}) in the last disjunct on memory reads indicate that
the memory access on field \form{\textit{val}} must  happen
before the one on field \form{\textit{next}}.

\paragraph{Invariant Generation}
After CHCs are generated,
{\toolname} automatically infers for
each inductive predicate an over-approximate invariant using the abstract interpretation technique
shown in \cite{Loc:CAV:2016}. 
For example, the invariant of
\code{sll} is \form{i{\geq}0 {\wedge} \error{=}0}, which is generated through three steps.
First, {\toolname} introduces an unknown predicate \form{\code{P}(i{,}\res{,}\error)}
as a place-holder for the invariant 
 \form{\forall i{,}\res{,}\error {\cdot} \code{sll}(i{,}\res{,}\error) {\implies} \code{P}(i{,}\res{,}\error)}.
Secondly, in the induction step, it unfolds predicate \code{sll} in the left-hand
side prior to substituting  all
occurrences of the predicate \code{sll}
with the induction hypothesis above to obtain:
\begin{small}
\[
\begin{array}{l}
\form{(\emp {\wedge}\res{=}\nil {\wedge} i{=}0 {\wedge}\error{=}0 }
~{\vee}~  
\form{{\exists} r {\cdot} \sepnode{\res}{\textit{node}}{i{,}r} {\sep}} \form{
\seppred{\code{P}}{i{-}1{,}r{,}\error} {\wedge} i{\neq}0}) 
{\implies} \code{P}(i{,}\res{,}\error)
\end{array}
\] \end{small}
Lastly, it over-approximates
the heap (e.g., transform \form{\sepnode{\res}{\textit{node}}{i{,}r}} to \form{\res\neq\nil}):
\begin{small}
\[
\begin{array}{l}
\quad\form{(\res{=}\nil {\wedge} i{=}0 {\wedge}
\error{=}0 ~ {\vee}~{\exists} r {\cdot} \res{\neq}\nil {\wedge} \seppred{\code{P}}{i{-}1{,}r{,}\error} {\wedge} i{\neq}0)} 
~ \form{{\implies}~ \code{P}(i{,}\res{,}\error)}
\end{array}
\]\end{small}
This constraint is then passed
to a fixed point calculator (e.g., the one in \cite{Trinh:APLAS:2013})
to
compute the closure form for \code{P}. Similarly, it generates invariant $0{\leq}\res{\leq}1$ for \code{main} and $0{\leq}\error{\leq}1$ for \code{loop_{11}}.
Our encoder uses these invariants
to prune any infeasible CHC whose body is unsatisfiable.
For instance,
while encoding the procedure \code{main},
the engine prunes the following CHC \form{\horn_1}:
\form{\exists \error_1 {\cdot}\seppred{\textit{sll}}{n{,}x{,}\error_1}_{\_}^0
 {\wedge}n{\geq}0 {\wedge}\error{=}1 {\wedge}\error{=}\error_1 {\ass}
 \seppred{\textit{main}}{n{,}\res{,}\error}}.
Due to the inconsistency between
the sub-formula  \form{\error{=}1} and the over-approximating invariant
of the predicate \form{\seppred{\textit{sll}}{n{,}x{,}\error}} with
the sub-formula \form{\error{=}0}
of the body of \form{\horn_1}, this CHC is reduced into \form{\horn_1{:} \false {\ass}
 \seppred{\textit{main}}{n{,}\res{,}\error}}. Thus, it could be discarded.

\paragraph{Query Generation}
After generating the CHC, we reduce the verification problem
into a decision problem i.e., whether there exists a feasible error path starting
from \code{main} or not.
Indeed, this encodes the ``liveness'' property, to ask whether ``something good eventually happens''
where ``good'' is an error.
Particularly to the example, the problem is encoded as
 \form{\D_0{\equiv} \seppred{\textit{main}}{n_0{,}\res_0{,}\error}^0_0 {\wedge}\error{=}1}.
The program contains an error if it contains a feasible {\em error} path i.e.,
there exists a satisfiable {\em base} formula
of the form \form{\base {\wedge} \textcolor{red}{\error{=}1}} (where \form{\base {\wedge}\error{=}1 \not\equiv\false})
derived from \form{\D_0}.
 Otherwise,
the program is \textcolor{blue}{safe} 
(i.e.,
every base formula derived from \form{\D_0} is equivalent to \form{\false {\wedge} \textcolor{red}{\error{=}1}}).
\subsection{Decision Procedure {\algonameesl}}\label{mov.solver}
%
Central to {{\algonameesl}} is the construction of a cyclic execution tree
based on which {{\algonameesl}} aims to decide whether \form{\D_0} is satisfiable.
Formally, a cyclic execution tree \form{\utree{i}} is 
 a tuple \form{(V, E, \backfun)} such that
 $V$ is a finite set of nodes representing formulae, which in turn represents execution paths;
 $E$ is a set of edges such that \form{(\D, \D') \in E}
 means
 that we can unroll a procedure (or a loop) in \form{\D},
 execute it and obtains \form{\D'} as the post path conditions;
and $\backfun$ is a back-link (partial) function which captures virtual cycles in the tree.
 A cycle, e.g. \form{\ctree{\D_c}{\D_b}{\sub}}, is in $\backfun$ if
the leaf \form{\D_b} is linked back to its ancestor \form{\D_c}
via a substitution \form{\sub} s.t.
 \form{\D_c ~{\equiv}~ \D_b\sub}.
In this back-link, \form{\D_b} is referred as a {\em bud}
and \form{\D_c} a {\em companion}.
Given a formula (e.g., \form{\D_0}), we can systematically construct a cyclic execution
tree $(V, E, \backfun)$ such that $V$ contains only formulae
derived from the root.
A leaf node is marked open or closed.
A closed leaf in the tree is either an unsatisfiable formula,
or a satisfiable {\em base} formula 
\underline{satisfying} \form{\error{=}1},
 or one is linked back to an interior node.


Given \form{\D_0}, {{\algonameesl}} constructs a cyclic execution tree which 
either contains a closed leaf representing a satisfiable base formula or contains only closed leaves which are either unsatisfiable or linked back.
Let us illustrate this via our earlier example.

\noindent{\bf Example Revisited}
{{\algonameesl}} starts with the encoding of  \code{main}:
\form{\D_0{\equiv} \seppred{\textit{main}}{n_0{,}\res_0{,}\error}^0_0  {\wedge}\error{=}1
}.
The
 execution tree of \form{\D_0} is shown in Fig.~\ref{fig.unfold.tree}, where
 the unsatisfiable nodes are underlined.
\form{\D_0} has two children, \form{\D_{11}} and \form{\D_{12}}, 
obtained by unrolling predicate \form{\textit{main}}.
\saveone\[\saveone
\begin{small}
\begin{array}{rclrcl}
\form{\D_{11} & {\equiv} & \emp {\wedge}n_0{<}0 {\wedge} \res_0{=}0 {\wedge} \underline{\error{=}0} {\wedge} \underline{\error{=}1}
} &\quad
\form{\D_{12} & {\equiv} & \seppred{\textit{sll}}{n_0{,}x_0{,}\error_1}_{1}^0
{\sep} \seppred{\textit{loop}_{11}}{x_0{,}x{,}\error_2}_{1}^1 {\wedge}n_0{\geq}0 {\wedge}\pure_{\error}} 
\end{array}
\end{small}
\]
\begin{wrapfigure}{l}{0.52\textwidth}
\savespace  \savespace 
\begin{scriptsize}
\begin{center}

\begin{tikzpicture}[node distance=18mm,level 1/.style={sibling distance=62mm},
      level 2/.style={sibling distance=27mm},
                        level distance=22pt, draw]
  \tikzstyle{every state}=[draw,text=black]

\node (A)                    {{$\D_0$}};
  \node         (B1) [below left=4mm and 25mm of A] {\underline{$\D_{11}$}};
  \node         (B2) [below =4mm of A] {\textcolor{blue}{$\D_{12}^\bigstar$}};
  \node         (D1) [below left=3mm and 10mm of B2] {{$\D_{21}$}};
  \node         (D2) [below right=3mm and 10mm of B2] {{$\D_{22}$}};
  \node         (F1) [below left=3mm and 5mm of D1] {\underline{$\D_{31}$}};
  \node         (F2) [below =3mm of D1] {\underline{$\D_{32}$}};
  \node         (F3) [below right=3mm and 5mm of D1] {\underline{$\D_{33}$}};
   \node         (F4) [below left=3mm and 5mm of D2] {\underline{$\D_{41}$}};
  \node         (F5) [below =3mm  of D2] {\underline{$\D_{42}$}};
  \node         (F6) [below right=3mm and 5mm of D2] {\textcolor{blue}{$\D_{43}^{\bigstar}$}};
  ;

  \path (A) edge              node {} (B1)
            edge              node {} (B2)
        (B2) edge              node {} (D1)
            edge              node {} (D2)
        (D1) edge              node {} (F1)
            edge              node {} (F2)
            edge              node {} (F3)
        (D2) edge              node {} (F4)
            edge              node {} (F5)
            edge              node {} (F6)
        (F6) edge [->,bend right=45,dotted]  node {} (B2)
            ;
\end{tikzpicture}
\caption{Execution Tree $\utree{}$.} \label{fig.unfold.tree}
\end{center}
\end{scriptsize}
\savespace  \savespace 
\end{wrapfigure}%
where \form{\pure_{\error} {=} \error=\error_2 {\wedge}
 \error{=}1 }.
For simplicity, existentially quantified variables are skolemized and path traces (i.e., the labels to mark program paths for witness generation) are discarded. For a flow-sensitive analysis,
the unfolding number of  predicates \form{\textit{sll}} and \form{\textit{loop}_{11}}
in \form{\D_{12}} is set to be greater than
 the unfolding number
of predicate \form{\textit{main}} in \form{\D_0}, and its order
is set to be the sum of the order of the predicate
in the definition of \code{main} (i.e., 0 for \form{\textit{sll}}
and 1 for \form{\textit{loop}_{11}}) and  the order of predicate \textit{main}
in \form{\D_0} (i.e., 0).
\form{\D_{11}} is marked closed as it is  unsatisfiable. The unsat core of \form{\D_{11}}
is underlined above. \form{\D_{12}} is then expanded to obtain two children, \form{\D_{21}} and \form{\D_{22}}, 
obtained by unrolling
 predicate \form{\textit{sll}}.
%
\[
\begin{small}
\begin{array}{l}
\form{\D_{21} ~{\equiv}~ x_0{=}\nil {\wedge} n_0{=}0 {\wedge}\error_1{=}0
{\wedge} \seppred{\textit{loop}_{11}}{x_0{,}x{,}\error_2}_{1}^1 {\wedge} 
\,
n_0{\geq}0 {\wedge}} \form{ \error_1{=}0 {\wedge} \pure_{\error}} \\
\form{\D_{22} ~ {\equiv}~ \sepnode{x_0}{\textit{node}}{n_1{,}r_2} {\sep}
\seppred{sll}{n_1{,}r_2{,}\error_1}_{2}^0 {\sep}\seppred{\textit{loop}_{11}}{x_0{,}x{,}\error_2}_{1}^1}
 {\wedge} \form{n_0{\neq}0{\wedge} n_1{=}n_0{-}1 {\wedge}n_0{\geq}0
 {\wedge} \pure_{\error}}
\end{array}
\end{small}
\]
\form{\textit{sll}} is chosen for unfolding rather than \form{\textit{loop}_{11}}
as they both have the same unfolding number and the sequence number
of the former (0) is smaller than that of latter (1).
This ensures the flow-sensitiveness of the execution.
%
Similarly, \form{\D_{21}} has three children obtained by unrolling
predicate \form{\textit{loop}_{11}} as follows.
\saveone\[\saveone
\begin{small}
\begin{array}{ll}
\form{\D_{31} {\equiv} & \emp {\wedge} x_0{=}\nil {\wedge} n_0{=}0 {\wedge} \error_1{=}0
 {\wedge}
x_0{=}\nil  {\wedge}
  x{=}x_0 {\wedge}}
 \form{  \underline{\error_2{=}0} {\wedge}} \form{  n_0{\geq}0
{\wedge}
\underline{\pure_{\error}} 
} \\
%
%
\form{\D_{32}  {\equiv} & \emp {\wedge} \underline{x_0{=}\nil} {\wedge} n_0{=}0 {\wedge}\error_1{=}0
{\wedge}{\neg}\dangl{x_0}{\wedge} \underline{x_0{\neq}\nil} 
 {\wedge}} \form{\sepnodeR{x_0}{val}{v_1}{0} 
{\wedge} x{=}x_0 {\wedge} v_1{<}0{\wedge}n_0{\geq}0
{\wedge}
  \pure_{\error}}\\
\form{\D_{33}  {\equiv}& \seppred{\textit{loop}_{11}}{x_1{,}x{,}\error_2}_{2}^0{\wedge} \underline{x_0{=}\nil} {\wedge} n_0{=}0 {\wedge}\error_1{=}0
{\wedge}
 {\neg}\dangl{x_0}}\\
  &
 \form{ {\wedge} \underline{x_0{\neq}\nil} {\wedge}} 
\form{  \sepnodeR{x_0}{val}{v_1}{0} 
 {\wedge}v_1{\geq}0}{\wedge}
 \sepnodeR{x_0}{next}{x_1}{1} 
 {\wedge} \form{n_0{\geq}0
  {\wedge}  \pure_{\error}}
\end{array}
\end{small}
\]
 As \form{\D_{31}}, \form{\D_{32}}
 and \form{\D_{33}} 
 are unsat 
they are marked as closed. 
%
\form{\D_{22}} has three children obtained by unrolling
 predicate \form{\textit{loop}_{11}}. 
 \form{\seppred{\textit{loop}_{11}}{x_0{,}x{,}\error_2}_{1}^1}
is chosen but not \form{\seppred{\textit{sll}}{n_1{,}r_2{,}\error_1}_{2}^0} as the former has a
 smaller
unfolding number (i.e., 1) than the latter (i.e., 2).
\saveone\[\saveone\begin{small}
\begin{array}{rll}
\form{\D_{41} ~{\equiv} & \sepnode{x_0}{\textit{node}}{n_1{,}r_2} {\sep}
\seppred{\textit{sll}}{n_1{,}r_2{,}\error_1}_{2}^0
{\wedge} x_0{=}\nil  {\wedge} x'{=}x {\wedge}
  \underline{\error_2{=}0}
 {\wedge}n_0{\neq}0{\wedge} n_1{=}n_0{-}1}
 \form{
 {\wedge}n_0{\geq}0
{\wedge}  
\underline{ \pure_{\error}}
} \\
\form{\D_{42} ~{\equiv} & \sepnode{x_0}{\textit{node}}{n_1{,}r_2} {\sep}
\seppred{\textit{sll}}{n_1{,}r_2{,}\error_1}_{2}^0{\wedge}
{\neg}\dangl{x_0}
  {\wedge}x_0{\neq}\nil {\wedge}} \\
& \form{
\underline{
v_1{=}n_1
}
{\wedge} \underline{v_1{<}0} 
}
 \form{ {\wedge}\underline{n_0{\neq}0}}{\wedge}  \underline{n_1{=}n_0{-}1}
 {\wedge} \underline{n_0{\geq}0} {\wedge} 
 \form{ \pure_{\error}} \\
\form{\D_{43} ~{\equiv} & 
\seppred{\textit{sll}}{n_1{,}r_2{,}\error_1}_{2}^0
 {\sep}
\seppred{\textit{loop}_{11}}{x_1{,}x{,}\error_2}_{2}^1{\wedge}x_0{\neq}\nil {\wedge} v_1{\geq}0{\wedge}}\\
 &
 \form{
 \form{ v_1{=}n_1 
 {\wedge}
}  \form{x_1{=} r_2 
{\wedge}} n_0{\neq}0{\wedge} n_1{=}n_0{-}1 {\wedge}
 } \form{n_0{\geq}0
{\wedge}
}
  \pure_{\error}
\end{array}
\end{small}
\]
 \form{\D_{41}} and \form{\D_{42}} are unsatisfiable. 
 {\algonameesl} detects that 
 \form{\D_{43}} 
is \emph{subsumed} by \form{\D_{12}}. {\algonameesl} then adds a back-link from \form{\D_{43}} to
\form{\D_{12}} to form a {\em  cyclic proof}~\cite{Loc:CAV:2016}. Intuitively, this back-link means that although the path from \form{\D_{12}} to \form{\D_{43}} can be infinitely unrolled, there is no error detected this way and therefore we can stop unfolding \form{\D_{43}}. 
As the cyclic unfolding tree in Fig.~\ref{fig.unfold.tree} is closed,
 our system proves that the program is safe.


\section{CHC Encoding} \label{sec:encoding}

In this section, we present details on how we encode programs into CHC.
The syntax of
\begin{wrapfigure}{l}{0.6\textwidth}
   \saveone
\begin{center} 
\savespace\[\savespace
\begin{small}
\begin{array}{rcl}
\myit{prog} & ::= & \myit{datat}^* ~\myit{meth}^* \qquad
\myit{datat}  ::= \btt{data} ~c ~\{~ \myit{field;}^* ~\} \\
\form{\myit{proc} & ::= & t ~ \myit{mn}~((t~v)^*)~\{e\}} \\
 \myit{field} & ::= & t~v  \qquad 
 t ~::=~ c ~|~ \tau \qquad  
 \tau ~ ::= ~ \int ~|~ \bool \\ 
e & ::= & \code{NULL}~| 
 k^{\tau}~| 
 v~| 
 v~{:=}~e ~| 
 t~v~|~ 
 \myit{mn}(v_1,..,v_n)~|~ \\
& & 
 \code{assume({\form{\pure}}}) \mid 
 v{\der}f ~|~ 
 \btt{new}~c(v_1,..,v_n) ~|~ 
 \btt{free}~v ~|~\\ 
& & v_1{\der}f {:=} v_2 ~|~ e_1;e_2 ~| 
 \myif{v}{e_1}{e_2}   
\end{array}
\end{small}
\] 
\vspace*{-4mm}
\caption{A core language ($t$: a type; $v$: a variable)}\label{fig.syntax} 
\end{center} 
\savespace
\end{wrapfigure}

 a core program $\myit{prog}$ is defined in Figure~\ref{fig.syntax}. 
A program consists of multiple data structures and procedures. Each data structure $\myit{datat}$ is composed of multiple fields
each of which $\myit{field}$ is composed of a type $t$ and a name $v$. 
Each procedure $\myit{proc}$ is composed of a return type $t$, a procedure name $mn$, multiple parameters (each of which has a type and a name) and a command $e$.
We assume there is a special procedure called \code{main} which is the entry point of the program.
We note that our core language does not include loop
as all \code{while} loops have automatically been transformed into
 {\em tail}-recursive procedures by the frond-end.


Our encoding is built upon forward symbolic execution for separation logic.
The proposed engine takes a core program as input, symbolically executes
it in a bottom-up manner, 
and produces  a system of CHC.
The execution rule for each command \code{e} is formalized as
 follows:
\begin{small}{\htriple{\D_1}{{\it e}}{\D_2}}\end{small}
where \form{\D_1} and \form{\D_2} are its precondition and postcondition, respectively.
%
A procedure \code{mn} is encoded by the predicate \form{\textit{mn}} as:
\hide{
\[\saveone
\begin{small}\hlrnone{METH}{ V{=}\{v_1{,}..{,}v_n{,}\res{,}\error\} \quad
 \htriple{\emp{\wedge}\nochange(V){:}[]}{e}{\bigvee \D_i{:}l_i} \\
  \myit{W}{=}\primeV(V) \quad \defpred~ \seppred{ \myit{mn}}{v_1{,}..{,}v_n{,}\res{,}\error}
 {\equiv} {(\bigvee ~{\exists} \myit{W}{\cdot}\D_i^e{:}l_i)\,} }{\vdash
t_0~\myit{mn}(t_1~v_1,..,t_n~v_n)~\{e\}}
\end{small}\]}
\saveone\[\saveone
\begin{small}\hlrnone{METH}{ V{=}\{\bar{v}{,}\res{,}\error\} \quad
 \htriple{\emp{\wedge}\nochange(V){:}[]}{e}{\bigvee \D_i{:}l_i} \quad
  \myit{W}{=}\primeV(V)  }{
t_0~\myit{mn}(\overline{t~v})~\{e\}
 \yields \defpred~ \seppred{ \myit{mn}}{\bar{v}{,}\res{,}\error}
 {\equiv} {(\bigvee ~{\exists} \myit{W}{\cdot}\D_i^e{:}l_i)\,}
}
\end{small}\]
 The list of arguments of the predicate \form{\textit{mn}}
includes the set of parameters of \code{mn}
and two preserved variables: \form{\res} for
output
and \form{\error} for error status (initially \form{\error{=}0}).
Each program path in \code{mn} corresponds
 to a disjunct constituting the predicate  \form{\textit{mn}}.
Function
$\primeV(V)$ returns the fresh form of variables in
\sm{V} which capture the symbolic values of the inputs.
$\nochange(V)$ returns equalities that assign
variables in $V$ to the symbolic values in $\primeV(V)$.
Furthermore, for the symbolic execution of each procedure, we assume that
the engine maintains a pair of two numbers \form{n_l}
and \form{k}, both initially assigned to \form{0},  where
 \form{n_l} is the next procedure call number
(used to preserve the order the function calls),
and  \form{k}
for the next memory access sequence number
 (used to preserve the order of the memory
accesses).
 Every  pointer parameter ${v_i}$
is initially set to be dangling, with the constraint \form{\dangl{v_i}}
conjoined 
with the precondition before executing.
Part of the rules for forward symbolic execution are presented in Fig. \ref{fig.ntyperules}, the
remaining rules are shown in the App. 
The engine halts when it detects an error in the precondition (rule \rulen{ERR}). 
Afterwards, it produces a disjunctive formula
(\form{\bigvee \D_i}) that
precisely captures the post-states of the function
where each \form{\D_i} is a program path in the procedure. 

%

\begin{figure*}[t]
{\scriptsize \begin{center}
 \begin{minipage}{\textwidth}
 \begin{frameit}\vspace*{-2.5mm}\savespace
 
\[
\hlr{NEW}{}
{\htriple{\D{:}l}{\new~c(\setvars{v})}{(\D\,{\sep}\,\sepnodeF{\res}{c}{\setvars{v}}){:}l}}
\quad
\hlrside{CALL-ERR}{
    }
{\htriple{\D{:}l}{mn(\setvars{x})}{\D {\sep}  \seppred{ \myit{mn}}{\setvars{x}{,}\res{,}\error}^{n_l} {\wedge}\error{=}1 {:}l}}{\form{\dagger}}
\]
\[
\hlr{ASSIGN}{\htriple{\D{:}l}{e}{\D_1{:}l_1}}
 {\htriple{\D{:}l}{v{:=}e}{\exists \res{,}
v'{.}~(\D_1[v'/v] {\wedge} v{=}\res){:}l_1}}
~~
 \quad
\hlrside{CALL-OK}{
  }
{\htriple{\D{:}l}{mn(\setvars{x})}{\exists v'{.}~\D[v'/\error] {\sep}  \seppred{ \myit{mn}}{\setvars{x}{,}\res{,}\error}^{n_l} {:}l}}{\form{\dagger}}
\]
\[
\hlr{FREE-ERR1}{
}
{\htriple{\D{:}l}{\code{free}~v}{(\D[\error'/\error]{\wedge}v{=}\nil{\wedge}\error{=}1){:}l}}
\qquad
\hlr{FREE-ERR2}{
}
{\htriple{\D{:}l}{\code{free}~v}{(\D[\error'/\error]{\wedge}\dangl{v}{\wedge}\error{=}1){:}l}}
\]
\[
\hlr{ERR}{
}
{\htriple{\D\wedge\error{=}1{:}l}{e}{\D\wedge \error{=}1{:}l}}
\quad 
\hlrside{FREE-OK}{
}
{\htriple{\D{:}l}{\code{free}~v}{ (\D[v'/v]{\wedge}{\neg\dangl{v'}}{\wedge}v'{\neq}\nil{\wedge}\sepnodeD{v'}{k}
{\wedge}\dangl{v} ){:}l}}
{\form{\ddagger}}
\]
\[
\hlr{LOAD-ERR1}{
}
{\htriple{\D{:}l}{v{\der}f_i}{(\D[\error'/\error]{\wedge}v{=}\nil{\wedge}\error{=}1){:}l}}
\qquad
\hlr{LOAD-ERR2}{
}
{\htriple{\D{:}l}{v{\der}f_i}{
    (\D[\error'/\error]{\wedge}{\dangl{v}}{\wedge}{\error{=}1} ){:}l}}
\]
\[
\hlrside{LOAD-OK}{
}
{\htriple{\D{:}l}{v{\der}f_i}{
     (\D{\wedge}  {\neg\dangl{v}}{\wedge}v{\neq}\nil{\wedge}
\sepnodeR{v}{f_i}{\res}{k} ){:}l}}{\form{\ddagger}}
\]
\[
\hlr{STORE-ERR1}{
}
{\htriple{\D{:}l}{v{\der}f_i{:=}v_0}{(\D[\error'/\error]{\wedge} v{=}\nil{\wedge}\error{=}1){:}l {:}l}}
\qquad
\hlr{STORE-ERR2}{
}
{\htriple{\D{:}l}{v{\der}f_i{:=}v_0}{
(\D[\error'/\error]{\wedge}{\dangl{v}}{\wedge}{\error{=}1} ){:}l }}
\]
\[
\hlrside{STORE-OK}{
}
{\htriple{\D{:}l}{v{\der}f_i{:=}v_0}{ (\D{\wedge} {\neg\dangl{v}}{\wedge}v{\neq}\nil{\wedge}
\sepnodeW{v}{f_i}{v_0}{k} ){:}l}}
{\form{\ddagger}}
\]
 \end{frameit}
 \end{minipage} \savespace
\caption{Encoding CHC where \form{v'},
\form{\error'} are fresh, $\dagger{:}$ \text{ increase }\form{n_l}, and $\ddagger{:}$ \text{ increase }\form{k}.
}\label{fig.ntyperules}
\end{center}}
\end{figure*}

A function/procedure call \myit{mn} is encoded by
two disjoint cases (the rules \rulen{[CALL-OK]} for normal
returning call and \rulen{[CALL-ERR]} for error returning call)
 through
an occurrence of the inductive predicate
with an increasing order number \form{n_l}.
This  predicate occurrence
is spatially conjoined (\form{\sep}) into the precondition of the triple.
The correctness of this spatial conjunction \form{\sep} is as follows.
If \form{\myit{mn}} allocates new heaps, 
 the spatial conjunction states correctly the separation between the new heaps
and the existing heap region in the precondition.
Otherwise, \form{\myit{mn}} includes only pure constraint (e.g.,
 \form{\emp{\wedge}\pure}),
the correctness is ensured by
the axiom:
 \form{(\heap_1{\wedge}\pure_1){\sep}(\emp{\wedge}\pure)
\Leftrightarrow  \heap_1{\wedge}\pure_1{\wedge}\pure} \cite{Reynolds:LICS02}.
 \form{n_l} captures the order of function calls
and is essential for {\em flow-sensitive} analysis. The flow-sensitiveness is ensured by
the decision procedure {\algonameesl} such that
{\algonameesl} uses this number to choose an inductive predicate
for unfolding in a breadth-first manner.
The encoding of memory accesses is one of our main contributions.
\code{free} is encoded in a way such that it aids for both safety and double free error discovery.
Essentially, \form{\code{free}~v} is erroneous if
the precondition \form{\D} implies  either
\form{v} is
 a dangling pointer (\form{\dangl{v}}) or
\form{v} has been assigned to \code{NULL} (\form{v{=}\nil}).
\form{\code{free}~v}  is safe if
\form{\D} implies (\form{\sepnodeF{v}{c}{...}}).
The key point to detect double free errors is that the
post-condition of the command \form{\code{free}~v}  must ensure
that the latest value of \form{v}
 points to a dangling heap, i.e. \form{\dangl{v}}.


For memory safety of free and memory access commands (load and store)
 over the heaps, we delegate the binding to the
 logic layer where symbolic path traces include explicit points-to predicates.
In the cases of memory safety, we add 
predicates with explicit binding notation, \form{\sepnodeR{v}{f_i}{x}{k}}
for memory read and
 \form{\sepnodeW{v}{f_i}{x}{k}} for memory write,
into program states
and postpone the binding to the normalization
 of the error/safe decision procedure.
 (We recap that \form{k} is a globally increasing order
among
 load/store/free operations over the heaps.)
Constraints used to encode safety and 
pointer dereference errors
are analogous to the encoding of the \code{free} command.

\section{Decision Procedure} \label{sec:solving}
{\algonameesl},
the decision procedure in {\toolname},
 is presented in Algorithm~\ref{algo.s2.sat}.
It takes a formula $\D$ and a place holder \code{[\,]} for a counterexample as inputs. It has two possible outcomes: 
 {\sat} with a counterexample \form{\xi} (i.e., a sequence of statements
starting from the entry point to the error statement)
 or {\unsat} with a cyclic proof. 
 If it does not terminate after some threshold time units, we mark the result as \code{UNKNOWN}.

At line 1, an execution tree with only one root node $\D$ is constructed. The loop (lines 2-11) then iteratively grows the tree while attempting
to establish a {\sat} proof or a {\unsat} proof.
In particular, at line 3, it checks whether
there is an error or not using a function \code{\uosl}(\utree{}).
 Function {\uosl} focuses on those leaf nodes which do not contain inductive predicates (referred to as base formulas).
If one of the leaf nodes is proved to be satisfiable, the algorithm returns {\sat} (line 4) together with the satisfiable node
 as a counterexample (obtained by using function \code{get\_error}). Otherwise, those leaf nodes
are marked as closed.
 At line 5, function {\oosl} prunes infeasible execution traces.
Afterwards, every remaining leaf node is checked to see if it can be linked-back through function {\lbsl}.
At line 7, it checks whether the tree is closed (i.e. all leaf nodes are either marked closed or can be linked-back) and returns {\unsat} if so.
Otherwise, at line 9, it chooses a leaf which is neither marked closed nor linked-back and grows the tree by unfolding the inductive predicate in the node using function \code{\unfoldsl}. Afterwards, the same process is repeated until either {\sat} or {\unsat} is returned or a timeout occurs.
In the following, we describe \code{\oosl}, \code{\unfoldsl}. \code{\uosl} and \code{\lbsl}  in more detail.

 \begin{algorithm}[t]\begin{small}
  \SetKwInOut{Input}{input}\SetKwInOut{Output}{output}
  \SetAlgoLined

 \Input{($\D${:}[])} \Output{{(\sat, $\xi$)} or {\unsat}}
   $\utree{} {\leftarrow} (\D {:}[]) $
\tcc*{initialize root}

  \While { $\true$ }
 {
    (\code{is\_error},\utree{}) {$\leftarrow$} \code{\uosl}(\utree{})
 \tcc*[r]{base case}

   {\bf if} {\code{is\_error}} {\bf then }
   {\Return{(${\sat}$, \code{get\_error}(\utree{}))}}; 

     $\utree{} {\leftarrow} \code{\oosl}(\utree{})$
\tcc*[r]{prune infeasible}

      $ \utree{}  {\leftarrow} \code{\lbsl}(\utree{}) $
\tcc*[r]{induction}

      {\bf if} {\code{is\_closed}(\utree{})} {\bf then } {\Return{{\unsat} };} 

	 \Else{

       $ \utree{}  {\leftarrow} \code{\unfoldsl}(\utree{}) $\tcc*{expand}
      }
 }
 \caption{Decision Algorithm ~{\algonameesl}}\label{algo.s2.sat} 
\end{small}\end{algorithm}


\paragraph{Function {\oosl}}
\code{\oosl} checks
unsatisfiability of every open leaf node using {\algonamesl}. 
First, {\algonamesl} 
replaces every
inductive predicate
by its over-approximated {\em invariant}
to obtain bases formulas.
If the base formula does not contain an (over-approximated) error
(e.g., \form{\error{=}1}), it checks whether
the base is unsatisfiable and marks it closed if it is.

\paragraph{Function {\unfoldsl}} {\unfoldsl} first finds an
open leaf node \form{\D} in the tree
 in a breadth-first order.
 If \form{\D} contains multiple occurrences of inductive predicates,
the occurrence \form{\seppred{\code{P}}{\setvars{v}}_u^o}
with the smallest \code{u} (i.e. the number of unfoldings) is selected.
If there are more than one occurrences that have the same
smallest number of unfoldings, the one with the smallest \code{o} is selected.
After that, it unfolds \form{\seppred{\code{P}}{\setvars{v}}_u^o} by spatially conjoining
branches of the instance \form{\seppred{\code{P}}{\setvars{v}}} into the formula.
This step also combines the path traces of the branches with the present path trace.
For example,
suppose \form{\D} be
 \form{{\exists} \setvars{w}_0 {.}~ \heap {\sep} \seppred{\code{P}}{\setvars{v}}_u^o {\wedge} \pure}
and \form{\seppred{\code{P}}{\setvars{v}} \equiv {\bigvee} (\D_i~\trace_i)}, then the unfoldings is  \form{{\exists} \setvars{w}_0 {.}~ \heap {\sep} {{\bigvee} (\D_i~\trace_i)} {\wedge} \pure}. The set of new leaves are: $\{{ ({\exists} \setvars{w}_0 {\cdot} \heap {\sep} \D_i {\wedge} \pure){:}
 (l{::}(\code{P}{,}\trace_i))} \}$.
In the illustrative example, the interprocedural path
corresponding to
the leaf \form{\D_{33}} 
 is \code{[(main, [2]); (sll,[1]);[(loop_{11},[1;1])]]}.
This trace captures the program path from the entry of the procedure \code{main}
to the \code{else} branch, call function \code{sll} once (\code{then} branch)
and to the body of the \code{while} loop once (\code{else} branch).
%
For a flow-sensitive analysis,
 the unfolding and order numbers are updated as follows.
Let \form{\code{Q}(\setvars{t})^{o_l}} be an occurrence of
an inductive predicate in a branch
of the definition of \code{P}. Its unfolding
number is set to \code{u{+}1} and
its order number to \form{o_l{+}o}.

\paragraph{Normalization}
Finally, the unfolded formula is normalized to
 remove all \code{ST}/\code{LD}/\code{DEL}
predicates. 
First, it detects and prunes infeasible cases. 
Based on the semantics described in Sect. \ref{sec:assert},
  we propose axioms on the memory access predicates
  as follows.
  \begin{small}
 \saveone \[\saveone
\begin{array}{lll}
\sepnodeR{v}{f}{\anon}{\anon} {\wedge} \dangl{v}   \Leftrightarrow \false &
~~ \sepnodeR{v}{f}{\anon}{\anon}  {\wedge} v{=}\nil  \Leftrightarrow \false &
~~ \sepnodeD{v}{\anon} {\wedge} v{=}\nil   \Leftrightarrow \false
\\
\sepnodeW{v}{f}{\anon}{\anon} {\wedge} \dangl{v}  \Leftrightarrow \false  &
~~ \sepnodeW{v}{f}{\anon}{\anon} {\wedge} v{=}\nil   \Leftrightarrow \false &
~~\sepnodeD{v}{\anon} {\wedge} \dangl{v}  \Leftrightarrow \false 
\end{array}
 \]\end{small}
After that, memory accesses are replayed over
the remaining satisfiable paths.
 Given a sequence of \code{ST}/\code{LD}/\code{DEL}
 memory predicates updating on the same field of a pointer,
it first sorts them based on  the ordering (i.e., \form{k}) numbers
and then performs the following normalization
 in the increasing order of the \form{k} numbers.
For a load \form{\sepnodeR{x}{f_i}{y}{k}}, it binds \form{y}
to the variable corresponding to the field \form{f_i}
of the points-to predicate \form{x}: 
\saveone\[\saveone\begin{small}
\begin{array}{l}
{\exists} \setvars{w} {.}\heap {\sep} \sepnodeF{x}{c}{f_1{:}v_1{,}..{,}f_i{:}v_i{,}...{,}f_n{:}v_n}  {\wedge}  {\neg\dangl{x}}{\wedge}
 \sepnodeR{x}{f_i}{y}{k} {\wedge} \pure\\
\qquad  {\Leftrightarrow}~{\exists}  \setvars{w} {.}\heap {\sep} \sepnodeF{x}{c}{f_1{:}v_1,..,f_i{:}v_i,...f_n{:}v_n} {\wedge}
 y{=}v_i {\wedge} \pure
\end{array}
\end{small}
\]
If the points-to predicate \form{\sepnodeF{x}{c}{..}} does not exist, either \form{\dangl{x}{\wedge}{\error{=}1}}
or \form{x{=}\nil{\wedge}{\error{=}1}} must exist
and hence,
 \form{\D} must have been marked as closed already. 
%
%
To eliminate 
\form{\sepnodeW{x}{f_i}{y}{k}}, we generate
a fresh variable \form{v_i'} and then bind \form{y}
to \form{v_i'}
 corresponding to the field \form{f_i}
of the points-to predicate \form{x}: 
\begin{small}
\form{
{\exists}  \setvars{w} {.}\heap {\sep} \sepnodeF{x}{c}{f_1{:}v_1,..,f_i{:}v_i,...f_n{:}v_n} {\wedge}{\neg\dangl{x}} {\wedge}
\sepnodeW{x}{f_i}{y}{k} {\wedge} \pure}
\form{\Leftrightarrow} \form{~{\exists}  \setvars{w} {.} (\heap {\sep} 
 \sepnodeF{x}{c}{f_1{:}v_1,..,f_i{:}v_i,...f_n{:}v_n} {\wedge} \pure) [v_i/v_i']
 {\wedge}
 v'_i{=}y }.
\end{small}
Similarly, a deallocation is performed as
\begin{small}
\form{
{\exists}  \setvars{w} {.}\heap {\sep} \sepnodeF{x}{c}{f_1{:}v_1,..,f_i{:}v_i,..,f_n{:}v_n} {\wedge}
 \sepnodeD{x}{k} {\wedge} {\neg\dangl{x}} {\wedge} \pure 
 {\Leftrightarrow}~{\exists}  \setvars{w} {.} \heap {\wedge} \pure}
\end{small}.


\paragraph{Base Case} Given an unfolding tree \utree{}, \code{\uosl} 
discharges every base formula at  open leaf nodes.
%
 Let \form{\base} be a base formula.  
{\uosl} transforms
\form{\base} into a SMT formula through
a reduction function, called \form{\xpure}.
We emphasize, applying any of the reductions presented
in \cite{Brotherston:LISC:14,Loc:CAV:2016,Makoto:APLAS:2016,Gu:IJCAR:2016,Xu:CADE:2017,Le:CAV:2017}
which do not take dangling pointers into account would obtain
an over-approximated abstraction.
Let \form{\base} be
\form{\exists \setvars{w}{\cdot}~
\heap
{\wedge}}\form{ \pure_{ptr}{\wedge}
\pure_{a}} where \form{\pure_{ptr}} is a conjunction
of equalities and dangling predicates over pointer-typed
variables and \form{\pure_{a}} is an arithmetical constraint.
Let \form{\setvars{v}} be the set of
variables
in \form{\base}. We assume that all these variables
are sorted. And \form{M(\setvars{v}_i)}
is a unique integer variable corresponding
to \form{\setvars{v}_i}, the \form{i^{th}} variable in \form{\setvars{v}}.
\form{\pure\equiv \xpure(\base)} is computed as:
(1) if
 \form{\sepnodeF{\setvars{v}_i}{c}{\anon} \in \heap} or
\form{\dangl{\setvars{v}_i} \in \pure_{ptr}}, then
 \form{M(\setvars{v}_i){=}i \in \pure}; 
(2) if \form{\setvars{v}_i{=}\nil \in \pure_{ptr}}, then
\form{M(\setvars{v}_i){=}0 \in \pure};
(3) if \form{\sepnodeF{\setvars{v}_i}{c}{\anon} \sep
\sepnodeF{\setvars{v}_j}{c'}{\anon} \in \heap}, then
\form{M(\setvars{v}_i){\neq} M(\setvars{v}_j)  \in \pure};
(4) if \form{\sepnodeF{\setvars{v}_i}{c}{\anon} \in \heap}
and \form{\dangl{\setvars{v}_j} \in \pure_{ptr}}, then
\form{M(\setvars{v}_i){\neq} M(\setvars{v}_j)  \in \pure};
(5) if \form{\setvars{v}_i{=}\setvars{v}_j \in \pure_{ptr}}, then \form{\setvars{v}_i{=}\setvars{v}_j  \in \pure}. Similarly to
\form{\setvars{v}_i{\neq}\setvars{v}_j};
and \form{\pure_{a} \in \pure}.
(1) captures the semantics of the domain of the heaps;
(2) asserts non-null heap addresses;
(3) and (4) capture the semantics of the separation
of heap locations; (5) and (6) forwards
 the constraints on the stack holding by \form{\base}.
By induction, we show the following lemma.
\begin{lemma}\label{lem.sat.base}
Given \form{\sstack}, we have
\begin{small} \form{\sstack \force \code{\xpure}(\form{\base})} \end{small}iff there exists  \begin{small}\form{\sheaps}\end{small}, such that 
\begin{small}\form{\sstack,\sheaps \force \base}\end{small}.
\end{lemma}

As \code{\xpure}(\form{\base}) is in the Presburger arithmetic,
 satisfiability
of \code{\xpure}(\form{\base}) is decidable. Furthermore,
\code{\xpure}(\form{\base}) can be discharged efficiently
by an SMT solver. 
%


\hide{\paragraph{Satisfiability for Inductive Formulas}
For efficiency, \code{\uosl} additionally discharges the unsatisfiability of
a formula with inductive predicates
through two steps as follows.
First, it places all occurences of inductive predicates by
their corresponding invariants to obtain a base formula. After that,
it invokes the procedure \code{\xpure} above to check the
satisfiability. If the outcome is unsatisfiable, it marks the leaf as closed.}

\paragraph{Cyclic Proofs}\label{sat.ind}
Function {\lbsl} constructs back-links as follows.
\hide{This function relates to the following projections.

We first define numeric  and
 spatial
  projection functions.
 These projections are critical for computing
back-links of our decision procedure.
We write $\setvars{x}^N$ and $\setvars{x}^S$
 to denote
the sequence of integer variables
 and pointer variables
 in $\setvars{x}$, respectively.
 For each \form{c \in {\Dns}}, we assume
 node symbol \code{c^S} for its spatial projection.
 \code{c^S} is a sequence of fields similarly to \form{c}
 except that non-pointer fields are replaced by the \form{\_} fields.

\begin{defn}[Spatial Projection]\rm
\label{defn.spatial.proj}
The spatial projection
$(\constr)^S$ is defined inductively as follows.
\saveone\[\saveone
\begin{array}{ll}
\form{(\D_1 \vee \D_2)^S}  \equiv
 \form{(\D_1)^S} \vee \form{(\D_2)^S} &
\form{(\exists \setvars{x} \cdot \D)^S} \equiv
 \form{\exists \setvars{x}^S \cdot (\D)^S} \\
\form{(\heap{\wedge}\atom{\wedge}\phi)^S}  \equiv
\form{(\heap)^S{\wedge}\atom} &
\form{(\heap_1 {\sep} \heap_2)^S}  \equiv
\form{(\heap_1)^S} {\sep} \form{(\heap_2)^S} \\
\form{(\seppredF{\code{P}}{\setvars{v}})^S}  \equiv
\seppredF{\code{P^S}}{\setvars{v}^S} &
\form{(\sepnodeF{x}{c}{\setvars{v}})^S}   \equiv
 \form{\sepnodeF{x}{c^S}{\setvars{v}^S}}\\
\form{(\emp)^S}  \equiv  \emp \\
\end{array}
\]
\end{defn}
For each inductive predicate
\form{
\defpred~ \seppred{\code{P}}{\setvars{t}}  ~{\equiv}~
   \bigvee^n_{i=1} ({\exists} \setvars{w_i}{\cdot}~ \D_i:l_i);
},
 we assume the inductive predicate symbols
\code{P^S} and predicate \seppred{\code{P^S}}{\setvars{t}^S}
for its spatial projection satisfying
\savespace\[\savespace
 \form{
\defpred~ \seppred{\code{P}^S}{\setvars{t}^S}  ~{\equiv}~
   \bigvee^n_{i=1} ({\exists} \setvars{w_i}{\cdot}~ \D_i)^S:l_i;
}
\]

\begin{defn}[Numeric Projection]
\label{defn.numeric.proj}
The numeric projection
$(\constr)^N$ is defined inductively as follows.
\saveone\[\saveone
\begin{array}{ll}
\form{(\D_1 \vee \D_2)^N}  \equiv \form{(\D_1)^N} \vee \form{(\D_2)^N} &
\form{(\exists \setvars{x} \cdot \D)^N} \equiv
 \form{\exists \setvars{x}^N \cdot (\D)^N} \\
\form{(\heap{\wedge}\atom{\wedge}\phi)^N}  \equiv
\form{(\heap)^N{\wedge}\phi} &
\form{(\heap_1 {\sep} \heap_2)^N}  \equiv
\form{(\heap_1)^N} {\wedge} \form{(\heap_2)^N} \\
 \form{(\seppredF{\code{P}}{\setvars{v}})^N}  \equiv
 \seppredF{\code{P^N}}{\setvars{v}^N} &
\form{(\sepnodeF{x}{c}{\setvars{v}})^N}  {\equiv}
\form{(\emp)^N}  \equiv  \true \\
\end{array}
\]
\end{defn}
For each inductive predicate
\form{
\defpred~ \seppred{\code{P}}{\setvars{t}}  ~{\equiv}~
   \bigvee^n_{i=1} ({\exists} \setvars{w_i}{\cdot}~ \D_i:l_i);
},
 we assume the inductive predicate symbols
\code{P^N} and predicate \seppred{\code{P^N}}{\setvars{t}^N}
for its numeric projection satisfying
\savespace\[\savespace
 \form{
\defpred~ \seppred{\code{P}^N}{\setvars{t}^N}  ~{\equiv}~
   \bigvee^n_{i=1} ({\exists} \setvars{w_i}{\cdot}~ \D_i)^N:l_i;
}\]

\hide{For examples, the numeric projections of predicates \code{P}, \code{Q}
are as follows.

\[
\begin{array}{l}
\code{pred}~\form{\seppred{\code{P}^N}{n} \equiv
  {n{=}0: 0}}
~~{\vee}~ \form{~{\exists}n_1 {\cdot} \seppred{\code{Q}^N}{n_1}{\wedge} n{=}n_1:1}; \\
\code{pred}~\form{\seppred{\code{Q}^N}{n} \equiv
{\exists} ~n_1 {\cdot \seppred{\code{P}^N}{n_1}{\wedge}{n{=}n_1{+}1:0}}}
 \\
~~~ {\vee}~ \form{~{\exists} ~n_1 {\cdot}
\seppred{\code{Q}^N}{n_1}{\wedge}n{=}n_1{+}1:1} ;\\
\end{array}
\]}


\paragraph{Back-Link}}
For every {\em open} leaf node \form{\D^{bud}},
{\lbsl} checks whether there exists an interior node \form{\D^{comp}} such that
after some weakening,
\form{\D^{bud}} is subsumed by
\form{\D^{comp}}
and all base formulas in the subtree rooted by \form{\D^{comp}}
 are unsatisfiable.
Particularly, 
for global infinitary soundness \cite{Brotherston:APLAS:12,Loc:CAV:2016}, {\lbsl} only considers
those \form{\D^{bud}} and \form{\D^{comp}} of the following form.
\saveone\[\saveone\begin{small}
\begin{array}{rclrcl}
 \form{\D^{comp} & {\equiv} & \exists \setvars{w}_1{\cdot}\heap_{b_1} {\sep}
\seppredF{\code{P_1}}{\setvars{t}_1}_m^0{\sep} ...{\sep}\seppredF{\code{P_j}}{\setvars{t_i}}_m^i {\wedge} \pure_1} &\quad
\form{\D^{bud} & {\equiv} & \exists \setvars{w}_2{\cdot}\heap_{b_2} {\sep}
\seppredF{\code{P_1}}{\setvars{t}_1'}_n^0{\sep} ...{\sep}\seppredF{\code{P_j}}{\setvars{t_k'}}_n^k {\wedge} \pure_2}
\end{array}
\end{small}
\]
where \code{k{\geq}i}, \code{n{>}m}, \form{\heap_{b_1}} and \form{\heap_{b_2}} are base formulas. The constraint \code{n{>}m} implies
at least one inductive predicate is unfolded.
In particular, function \form{ {\lbsl}(\D^{bud}{,}\D^{comp})} is implemented as follows.
it finds a substitution \form{\sub}
such that:
 (a) for every \form{\sepnodeF{x_j}{c_j}{\setvars{v}_j} {\in} \heap_{b_1}}
there exists \form{\sepnodeF{x'_j}{c_j}{\setvars{t}_j} {\in} \heap_{b_2}}
and \form{\sepnodeF{x_j}{c_j}{\setvars{v}_j}\equiv \sepnodeF{x'_j}{c_j}{\setvars{t}_j}\sub};
(b) similarly,
every occurrence of inductive predicates in the companion is linked;
(c) \form{\pure_2\sub \imply \pure_1} holds.
\hide{free points-to predicates are exhaustively matched.
For every \form{\sepnodeF{x_j}{c_j}{\anon} {\in} \heap_{b_1}} and
\form{x_j {\notin} \setvars{w}_1},
there exists \form{\sepnodeF{x'_j}{c_j}{\anon} {\in} \heap_{b_2}}
s.t. \form{\pure_{2} \imply x_j'{=}x_j};
and vice versa, for every \form{\sepnodeF{x_j}{c_j}{\anon} \in \heap_{b_2}} and
\form{x_j {\notin} \setvars{w}_2},
there exists \form{\sepnodeF{x'_j}{c_j}{\anon} \in \heap_{b_1}}
s.t. \form{\pure_{1} \imply x_j'{=}x_j}.
\item Thirdly, occurrences of inductive predicates are exhaustively matched.
First, we replace every existentially quantified argument
of inductive predicates in \form{\heap_{b_1}} and \form{\heap_{b_2}}
by \form{\anon}.
After that for every \form{\seppredF{\code{P_j}}{\setvars{t}}_m^\anon \in \heap_{b_1}}, there exists \form{\seppredF{\code{P_j}}{\setvars{v}}_n^\anon \in \heap_{b_2}} s.t. for all \form{i \in \{1..n\}}
(where \form{n} is the arity of \code{P_j})
either \form{\setvars{t}_i}
is identical to \form{\setvars{v}_i} or both of them are
\form{\anon} and vice versa.
Note that this matching considers equality constraints in \form{\pure_1} and \form{\pure_2}.
\item Fourthly, (dis)equalities over pointers are matched.
\item Lastly, we check whether the pure constraint of
the bud implies the pure constraint of the companion.
Back-link is decided if this implication holds.}

\section{Implementation and Evaluation}\label{sec.impl}

{\toolname}  is programmed in OCaml. 
It uses
 Z3~\cite{TACAS08:Moura} for checking satisfiability of first-order logic formulae
produced by procedure \code{\xpure} for the base cases.
Moreover, the analyzer \btt{LFP} uses FixCalc~\cite{Popeea:ASIAN06}
 to compute an over-approximating invariant of an inductive predicate
to prune unsatisfiable nodes in advance via function \code{\oosl}.

In the presence of an erroneous trace, {\toolname} produces a counterexample
annotated with a stack record of a program path \form{\xi} leading to the error.
From the counterexample, we  extract the error location
and generate error witnesses \cite{Beyer:FSE:2015} in GraphXML \cite{DBLP:conf/gd/BrandesEHHM00}.
%
\begin{table*}[t]
\setlength{\tabcolsep}{1.4pt}
\begin{small}
\begin{center}
\caption{Experimental Results}
\savespace \savespace
\label{tbl:expr:rec:cmb}
\begin{tabular}[t]{|c | %
|c | c |c | c | c | c  | %
|c | c |c | c | c | c  |
| c | c| }
\hline
\multirow{2}{3em}{Tool} & \multicolumn{6}{|c|}{RQ1: Verification} & \multicolumn{6}{|c|}{RQ2: Falsification} &  \multicolumn{2}{|c|}{RQ3: Overall}\\
\cline{2-15}
  &  s$\surd$ &  s\ding{55} & unk & to &  pts & time
  & e$\surd$ & e\ding{55} &  unk & to &  pts & time %
  & pts & time
 \\
\hline
\rowcolor{Gray}{ SeaHorn \cite{Gurfinkel:CAV:2015}} & 59 & \textcolor{red}{11} & 0 & 4 & {-}58 & 14m37s
                                     &  57 & 0 & 0 & 1 & 57 & 4m8s
& {-}1 & 18m45s \\
{ Cascade \cite{Wei:VMCAI:2014}} &  0 & 0 & 58 & 16 &  0 & 48m45s
                                 & 8 & 0  & 50  & 0  & 8 & 13m5s
                                 & 8 & 61m50s \\
\rowcolor{Gray}{ PredatorHP \cite{Predator:SV:16}} & 39 & \textcolor{red}{1}  & 16 & 18 & 62 & 64m36s
                                   & 38  & 0 &  5 &  15 & 38 & 28m39s
 & 100 & 95m15s \\
{ ESBMC \cite{Cordeiro:2011:ICSE}} & 39 & \textcolor{red}{1} & 0 & 34 
                                                                        & 62 & 112m40s
                                                  & 55 & 0 & 0 & 3 
                                                                   & 55 & 11m33s
& 117 & 124m14s\\
\rowcolor{Gray}{ CBMC \cite{Clarke:TACAS:2004}} & 32  & 0 & 1 & 41 & 64 & 130m20s
                                 & 54 & \textcolor{black}{0} & 0 & 4  & 54  & 17m56s
                                 &  118 & 148m16s \\
{ Smack-Corral \cite{Haran:TACAS:2015}} & 44 & \textcolor{red}{1} 
                                                                   & 0  & 29  & 72 & 93m48s
                                                       & 50  & 0  & 0 & 8 & 50 & 31m48s
                               &  122 & 125m36s   \\
\rowcolor{Gray} { CPA-Seq1.9 \cite{Beyer:CAV:2007}} & 38  &  0 & 13 &  23 & 76 &  434m47s
                                 &  42 & 0 & 3 & 13 & 42 & 55m6s
                                 & 128 & 489m53s \\
{ UAutomizer \cite{Heizmann:CAV:13}} & 49 & 0 & 4 & 21 & 98 & 83m20s
                                     &  47 & 0  & 0 &  11 & 47 & 44m20s
& 145 & 127m40s \\
\rowcolor{Gray} {CPA-Seq1.4 \cite{Beyer:CAV:2007}} & 54  &  0 & 4 &  16 & 108 &  51m41s
                                 &  47 & 0 & 1 & 10 & 47 & 25m36s
                                 & 155 & 77m17s\\
{\bf \toolname} & {\bf 59} & {\bf 0} & {\bf 4} &  {\bf 11}& {\bf 118} & {\bf 37m31s}
                & {\bf 55} & {\bf 0} & {\bf 1} & {\bf 2} & {\bf 55 }  &  {\bf 7m31s }                & {\bf 173} & {\bf 45m2s} \\
\hline
\end{tabular}
\end{center}
\end{small}
\end{table*}

In the rest of this section, we conduct experiments to
answer  three research questions (RQ):
\begin{itemize}[leftmargin=*]\setlength{\itemsep}{0pt}
\item {\em RQ1}: Does {\toolname} verify programs correctly?
\item {\em RQ2}: Does {\toolname} falsify program properties correctly?
\item {\em RQ3}: Is {\toolname} more effective and efficient than existing tools?
\end{itemize}

\paragraph{Experimental setup}
We conduct the experiments on a set of challenging programs. 
Our test subjects include 132 C programs
taken from the SV-COMP benchmark repository~\cite{Dirk:SVCOMP:2016}, i.e.,
all 98 programs in the {\em Recursive} category, 24 programs  manipulating nested lists and binary trees 
in the
 {\em Heap Data Structures} category,
 and 10 
 generated by us.
	Among the 132 programs, 74 of them are safe whereas 58 are erroneous. 
The programs in the first category implement 
recursive algorithms like the towers of Hanoi, Fibonacci's numbers and McCarthy91.
The programs in the last category
 manipulate complex data structures (e.g., parallel lists, reversely doubly-linked list or tll trees),
 which
are yet to be included in~\cite{Dirk:SVCOMP:2016} due to their complexity.
%
Each program contains at least one {\em loop} and/or one {\em recursive} method call,
and has at least one pure assertion.
We set the upper bound of the unfolding number of all predicates in
{\toolname} to 28, i.e., after 28 unfoldings, if
the error/safe decision has not been made, {\algonameesl} returns \code{UNKNOWN}.

We compare our verification system {\toolname}
 with the state-of-the-art verification tools
participated in the SV-COMP competition including PredatorHP \cite{Predator:SV:16},
 SeaHorn \cite{Gurfinkel:CAV:2015}, Cascade \cite{Wei:VMCAI:2014},
ESBMC \cite{Cordeiro:2011:ICSE}, UAutomizer \cite{Heizmann:CAV:13},
CPA-Seq \cite{Beyer:CAV:2007}, CBMC \cite{Clarke:TACAS:2004} and
Smack-Corral \cite{Haran:TACAS:2015}.
Among the tools, SeaHorn is also a CHC-based verification system.
PredatorHP runs several instances of
 Predator \cite{Dudka:CAV:2011}
 in parallel
 and composes the results into a final verification verdict.
Cascade and 
ESBMC 
are based on Bounded Model Checking (BMC)
and were designed for falsification~\cite{Wei:VMCAI:2014,Herbert:SEFM:2015,Morse:TACAS:2013}.
Ultimate Automizer \cite{Heizmann:CAV:13} is a model checker based on automata.
CPA-Seq \cite{Beyer:CAV:2007} builds and refines reachability trees
combining multiple techniques like abstract interpretation, predicate abstraction, and
shape analysis.
Smack-Corral \cite{Haran:TACAS:2015}
translates LLVM compiler's intermediate representation into
SMT constraints.
For
CBMC, CPA-Seq1.9, ESBMC, PredatorHP
and Ultimate Automizer,
we apply the binaries and settings used in Sv-comp 2020 \cite{SVCOMP:2020}.
For the remaining tools, 
we apply the binaries and settings used in Sv-comp 2016~\cite{Dirk:SVCOMP:2016}.
All experiments were performed on a machine with an 
 Intel i7-6500U (2.5GHz)
2-processor and 8 GB RAM. 
The timeout is set to 180 seconds. 

\paragraph{Experimental Results}
The results are reported in Table \ref{tbl:expr:rec:cmb}.
The 1st column shows the verification systems.
The rest  shows results in relation to the three research questions.

\paragraph{RQ1}
To answer this question, we experiment on 74 safe programs.
The first two columns in this part show the number of correct safe (s$\surd$) and
incorrect safe (a.k.a. false positives) (s\ding{55}).
The next two columns show the number of unknown (unk column)
 and either timeout or crashes (to column).
We rank these tools based on their points (pts column).
 Following the same rules applied in the SV-COMP competition,
we gave +2 for one s$\surd$, -16 for one s\ding{55} and
 0 for either unknown or timeout or crash.
The last column shows the total wall time.
The results show that
SeaHorn and {\toolname} produced
the highest number of correct safe.
However, SeaHorn returned 11 incorrect answers, 9 of which are programs in the second category.
This incompleteness is due to the lack of a precise encoding
model for complex data structures.
PredatorHP, ESBMC and Smack-Corral reported one wrong outcome each.
Cascade did not return any incorrect results but it also did not produce any correct ones.
{\toolname} could not decide 15 programs, it returned
either \code{UNKNOWN} or
timeout, 12 of which are
in the first category and 3 are in the second one.
These programs contain 
mutual recursions, constraints outside the linear arithmetic,
or non-periodic recurrent relations.

\paragraph{RQ2}
To answer this question, we have experimented on 58 erroneous programs.
Columns in this part are similar to that of the safe programs, except that the first two columns show the number of
correct error (e$\surd$) and false negatives (e\ding{55}). Following
the same rules applied in the SV-COMP competition, we gave +1 for one e$\surd$ and -32 for one e\ding{55}.
The results show that CHC-based systems performed the best while SeaHorn produced the highest number of correct safe with the shortest running time. Nevertheless,
none returned a false negative. 
{\toolname} returned 3 timeouts
for programs requiring a very large number of
unfoldings. For instance, one of them is to compute
the Fibonacci number where the input is 25
and requires to invoke 242785 recursive calls.
This number of calls is beyond both our algorithm bound and running bound.

\paragraph{RQ3}
To answer this question, we sum the points and running time of tools
in the previous two parts and report in the respective
 two columns in the third part of Table \ref{tbl:expr:rec:cmb}.
The results show that {\toolname} is the most effective and ranks second in terms of efficiency.
Cascade, UAutomizer, CPA-Seq
and {\toolname} did not report any incorrect outcomes.
SeaHorn is the fastest in all two categories; however, it was the least in effectiveness.
PredatorHP is the third fastest
but reported a high number of false positives.
These experimental
 results imply that our proposal,
a combination of abstract interpretation, model checking and cyclic proof for CHC,
  is promising as
{\toolname} is effective and efficient.

\section{Related Work and Conclusion} \label{sec.related}

Closest to our work are CHC-based verification systems, e.g. Duality \cite{export:180055},
 HSF \cite{Grebenshchikov:PLDI:2012,Sergey:TACAS:2012}, SeaHorn \cite{Gurfinkel:CAV:2015} and
$\mu$Z~\cite{Hoder:CAV:2011}. 
Details on these systems based on CHC
were summarized in~\cite{Nikolaj:SMT:2012,DBLP:conf/birthday/BjornerGMR15}.
As mentioned above, these systems focus on numerical and array programs whereas we target heap-manipulating programs.
Our work also relates to cyclic proof verification systems 
\cite{Brotherston:POPL:08,Brotherston:CADE:17,Brotherston:CPP:17}.
While they support only safety verification (i.e. the absence of bugs),
ours supports reasoning about both safety and the presence of bugs.

Verification systems support for both verifying and falsifying
safety properties 
include
property-directed reachability (PDR) \cite{Itzhaky:CAV:2014,Aleksandr:CAV:2015}, interpolation \cite{Aws:ESOP:2015},
and separation logic \cite{Berdine:APLAS05,jacm.Calcagno11,Piskac:CAV:2013,Navarro:APLAS:2013,Jansen:ESOP:2017}.
PDR is a sound and relatively complete framework for shape analysis
using abstraction predicates. 
%
%
Smallfoot \cite{Berdine:APLAS05,jacm.Calcagno11} 
presents foundations in separation logic.
GRASShopper \cite{Piskac:CAV:2013}  and Asterix \cite{Navarro:APLAS:2013}
encode the verification conditions of heap-manipulating programs
 into an SMT supported form.
Our work shows, at the first time, how to employ cyclic proofs for
verifying and falsifying
the safety properties .
\hide{Moreover, VCs generated by the existing verification systems
based on separation logic rely on entailment checking
and their completeness is limited to simply linked data structures.
In contrast, VCs generated by our system are property-based queries;
the expressiveness of the decidable fragment which
includes complex data structures is promising to advance
the completeness of the verification system.}

Safety verification systems based on abstract interpretation
 have been developed for verification (not falsification), 
e.g.~\cite{Chin:SCP:12,CAV08:Yang,Dudka:CAV:2011,Holik:CAV:2013,Loc:CAV:2014,Loc:2018:S2ENT}.
These systems over-approximate input programs
to verify the absence of errors
and may produce false positives.
%
There are related techniques supporting loop-based programs, such as  PDR 
and interpolation (mentioned earlier), loop invariant generation
\cite{Corina:SPIN:2004,Furia:ACMS:2014,Li:ASE:2017},
loop precondition inference using cyclic proofs \cite{Brotherston-Gorogiannis:14},
k-induction based techniques \cite{Martin:SAS:2015,Dirk:CAV:2015}. 
In contrast to these, we transform
the source-code-based loops/recursions into logic-based inductive predicates and
rely on a decision procedure to search
for cyclic proofs to prune the subsumed execution paths 
for both safety verification and falsification.

Compare to the symbolic execution for separation logic
\cite{Berdine:APLAS05,Chin:SCP:12,Loc:NFM:2013,Muller:CAV:2016,Pham:ICSE:2018,Long:ATVA:2019,Long:FM:2019},
our system distinguishes itself
at the free command.
Without the dangling predicate, the operation semantics
of this command in these systems is: \form{\{\sepnodeF{x}{c}{\overline{v}}\}} \code{free~x} {\{\emp\}}.
That means, these systems over-approximate the post-state
 and
thus they may produce false positives.
We notice the incorrectness separation logic presented in
 \cite{ISL:2020} recently
which also makes use of negative heaps to reasoning about the presence of bugs.
Similar to ours, the operation semantics
of this command in this work is: \form{\{\sepnodeF{x}{c}{\overline{v}}\}} \code{free~x} {\{\dangl{x}\}}.
The main difference is that as
the proposal in \cite{ISL:2020} is to fix the over-approximation
of the frame rule, \hide{of separation logic,} their negative heap predicate
 is spatial based. As we use the dangling predicate to assert the guard
condition of the heaps, we treat them as pure constraints.
Moreover, while the engines in \cite{Berdine:APLAS05,Chin:SCP:12,Muller:CAV:2016}
reduces the compositional verification problem into the validity of  entailments,
the engine in \cite{Pham:ICSE:2018,Long:ATVA:2019,Long:FM:2019} generates test inputs from pre-conditions
to witness the reachability
in {\em bounded} programs,
 our engine is to transform
the verification problem into the reachability problem using CHC.

\hide{
This proposal also relates to those techniques supporting loop-based programs.
Besides PDR \cite{Itzhaky:CAV:2014,Aleksandr:CAV:2015}
and interpolation \cite{Aws:ESOP:2015},
while the proposals in
\cite{Cousot:POPL78,Henzinger:POPL:2004,Corina:SPIN:2004,Furia:ACMS:2014,Li:ASE:2017}
proposed to generate loop invariant, the work in \cite{Brotherston-Gorogiannis:14}
infers preconditions of loops to establish cyclic proofs,
the technique in \cite{Martin:SAS:2015,Dirk:CAV:2015} presented
k-induction to search for a strong enough inductive invariant 
to establish safety properties.
In contrast to these techniques, we transform
the source-code-based loops/recursions into logic-based inductive predicates and
rely on a decision procedure to search
for cyclic proofs to prune the subsumed execution paths 
for both safety verification and falsification.
}

Our solver ({\algonameesl}) 
 relates to 
decision procedures for satisfiability problems
in separation logic 
\cite {Berdine:APLAS05,NavarroPerez:PLDI:2011,Enea:APLAS:2014,Brotherston:LISC:14,Piskac:CAV:2013,Navarro:APLAS:2013,Gu:IJCAR:2016,Xu:CADE:2017,Loc:CAV:2016,Makoto:APLAS:2016,Le:CAV:2017,Le:APLAS:2018,Loc:VMCAI:2021}.
It is closest to  cyclic-based
solvers \cite{Loc:CAV:2016,Le:CAV:2017,Le:APLAS:2018,Loc:VMCAI:2021}.
The procedure in
\cite{Le:APLAS:2018} derived cyclic proofs for word equation problem and \cite{Loc:VMCAI:2021}
proposed to infer summary of inductive predicates for compositional satisfiability solving.
Compared with \code{S2SAT_{SL}} \cite{Loc:CAV:2016}, {\algonameesl} 
additionally supports 
dangling pointers 
and   memory  load/store/del predicates.
The decision procedure  in \cite{Le:CAV:2017} reduces
the satisfiability problem of an inductive predicate
into the satisfiability problem of its numeric part, which
requires a 
separation between the spatial and numeric parts.
In contrast to \cite{Le:CAV:2017} and similarly to \cite{Loc:CAV:2016}, our procedure
 creates back-links based on 
the weakening and subsumption.
Despite of their power, the decision procedures in \cite{Le:CAV:2017}
 can not be applied into this work
as the flow-sensitiveness forbids the separation of the two domains.

\paragraph*{Conclusion}
We have presented {\toolname}, a
CHC-based verification system using a novel extension of
separation logic.
Given a C program, our system transforms it into
CHC and invokes 
decision procedure {\algonameesl}
to symbolically execute and to discharge these CHC. {\algonameesl}
detects bugs with feasible paths,
or constructs a cyclic proof to prune subsumed safe paths
 (so as to prove bug-free programs).
We have implemented
a prototype and evaluated it on a set of heap-manipulating programs.
 The experimental results
shows that CHC-based approach is a promising approach
to software verification for both  the absence and the presence of bugs.


\bibliographystyle{abbrv}
\bibliography{all}


 \begin{subappendices}
   \renewcommand{\thesection}{\Alph{section}}%
   \section{Encoding CHC (Cont)}
   \[
\hlr{ASSUME}{\rho{=} \circ \{v'/v ~|~ v\in \FV(\pure) \}}
{\htriple{\D{:}l}{~\code{assume}~\pure}{(\D\rho{\wedge}\pure ){:}l}}
\quad
\hlr{SEQ}{
\htriple{\D{:}l}{e_1}{\D_1{:}l_1} \quad
\htriple{\D_1{:}l_1}{e_2}{\D_2{:}l_2}
}
    {\htriple{\D{:}l}{e_1;e_2}{\D_2{:}l_2}}
\]
\[
   \hlr{IF-TRUE}{
\begin{array}{c}
\htriple{\D{\wedge}v{:}(l{::}1)}{e_1}{\D_1{:}l_1}
\end{array}}
{\htriple{\D{:}l}{\myif{v}{e_1}{e_2}}{\D_1{:}l_1}}
 \qquad
\hlr{IF-FALSE}{
\begin{array}{c}
\htriple{\D{\wedge}{\neg}v{:}(l{::}2)}{e_2}{\D_2{:}l_2}
\end{array}}
{\htriple{\D{:}l}{\myif{v}{e_1}{e_2}}{\D_2{:}l_2}}
\]

The encoding steps for \code{assume}, \code{sequencing}, and \code{if}
commands are shown above.

   \section{Soundness and Termination}
\label{sat.term}

Given a program defined in our core language shown in Figure~\ref{fig.syntax}, the program has an error
(either \code{assume (\error{=}1)} is reachable or a heap-manipulation
is violated) if and only if there is a satisfiable error formula; and the program has no error if and only if there is a cyclic proof for the generated CHC.

\begin{theorem}
Given a core program P,
 (a) if
{\toolname} returns a bug, then there is some input
to
P
that leads to an error;
 (b) if {\toolname} terminates without a bug
 then there is no input that leads to a null pointer error or double free error;
 (c) otherwise,
{\toolname}  will run forever.
\end{theorem}

The correctness of the above Theorem relies on 
{\algonameesl}.
In the following, we analyze the soundness of the solver.
Inspired by 
~\cite{Loc:CAV:2016}, 
{\algonameesl} is designed for deciding a sound and complete base theory {\basetheory} augmented with inductive predicates. Furthermore, the {\em base theory} {\basetheory} must satisfy the following properties: (i)
{\basetheory} is closed under propositional combination and supports
boolean variables; (ii) there exists a complete decision procedure
for {\basetheory}.
In this work, the theory {\basetheory} is the fragment of base formulas.
The soundness and completeness of {\basetheory} is shown by Lemma \ref{lem.sat.base}.
\begin{lemma}
{\algonameesl} is sound when it returns either {\sat} or {\unsat}.
\end{lemma}
\noindent This lemma follows
Lemma \ref{lem.sat.base} and the soundness of cyclic proofs.

\paragraph{Termination}
As the number of procedures/loops in a program is finite,
our encoding is always terminating.
Therefore, {\toolname} terminates iff
{\algonameesl} terminates.
In the following, we show that
{\algonameesl} terminates in  the universal  fragment {\slap}
that is an extension of the decidable fragment
presented in \cite{Loc:CAV:2016}  with dangling predicates and the
heap-manipulation simulation predicates (\code{LD}, \code{ST}, and \code{DEL}).


{\slap} is a fragment of separation logic which supports all the syntax shown in Figure~\ref{prm.spec.fig}
 except that those definitions of inductive predicates are restricted as follows.
An inductive predicate \seppredF{\code{P}}{\setvars{t}} is in {\slap} if it is 
in the following form (\form{\setvars{v}_i {\subseteq} \setvars{w}} for all \form{i {\in} \{1...l\}}):
\[
\begin{array}[t]{l}
\form{\defpred~\seppredF{\code{P}}{\setvars{t}} {\equiv} \base_0 
\vee {\exists} \setvars{w}.~\Conv^n_{i{=}1}
\sepnodeF{x_i}{c_i}{\setvars{d_i}} {\sep} \Conv^n_{i{=}l}
 \seppredF{\code{P}}{\setvars{v}_i} {\wedge} \pure_r}
 \end{array}
 \]
where
\form{\base_0} is a  base formula, 
and \form{\pure_r} is a conjunction of ordering arithmetical constraints each of which
is of the form: either an arithmetical constraint over data fields \form{\setvars{d_i}} or \form{v_1{\leq}v_2{+}k}
or \form{v_1{\geq}v_2{+}k}, \form{k} is an integer number, \form{v_1{\in} \setvars{t} {\cup} {\bigcup^n_{i{=}1}}\{\setvars{d_i} \}} and \form{v_1\in  \bigcup^l_{i{=}1}\{\setvars{v_i} \}}.

\begin{theorem}\label{thm.term}
{\algonameesl} terminates for {\slap}.
\end{theorem}
%


To show that  {\algonameesl} always terminates for {\slap}, we essentially  prove that
for each leaf (respectively a bud) of the unfolding tree derived for an unsatisfiable formula
if it has not been classified as {\unsat} yet, it must be linked back to
an interior node (respectively a corresponding companion)
to form a cyclic proof in {\em finite} steps.
We remark that in the following, we only consider
 companions which are {\em descendant} of a bud.
We recap that a bud is linked to a companion
if free pointer variables of points-to predicates
and predicate instances are exhaustively matched
and fixed points of all numeric projections of buds are
computable.
In addition, the function {\lbsl} does not rely on
 constraints over data fields.
We use \form{\pure^{df}} to denote the constraint
on variables of data fields of points-to predicates. We show {\algonameesl} always terminates for {\slap} through three steps.
First, we show that
the termination of {\algonameesl}
does not rely on the pure constraints over data fields of points-to predicates.
\begin{lemma}\label{lem.arith}
Given a bud
\form{\D_{bud}{\equiv}\exists \setvars{w}_2\cdot\D_2{\wedge}\pure^{df}_2}
if there exists a companion
\form{\D_{comp}{\equiv}\exists \setvars{w}_1\cdot\D_1{\wedge}\pure^{df}_1}
such that
 the subformula \form{\exists \setvars{w}_2\cdot\D_2}
 can be linked back to the subformula
 \form{\exists \setvars{w}_1\cdot\D_1}, then \form{\D_{bud}}
can be linked back to \form{\D_{comp}}.
\end{lemma}


\begin{proof}
 We assume
 \form{\exists \setvars{w}_2\cdot\D_2} can be linked back to
\form{\exists \setvars{w}_1\cdot\D_1}.
As so, the heap in bud can be weakened and matched with the heap
in the companion with some substitutions \form{\rho}.
And \form{\pure^{df}_2} after weakened and substituted to become
  \form{\pure'^{df}_2},
at the last step we need to prove  \form{\pure'^{df}_2 {\implies} \pure^{df}_1}.
Suppose that \form{\pure'^{df}_2{\equiv} \pure'^{df}_{21} {\wedge} \pure'^{df}_{22}}
where \form{\FV(\pure'^{df}_{21}) {\in} w_2} and \form{\FV(\pure'^{df}_{22})}
are observable variables.
Similarly, suppose that \form{\pure'^{df}_1{\equiv} \pure'^{df}_{11} {\wedge} \pure'^{df}_{12}}
where \form{\FV(\pure'^{df}_{11}) {\in} w_1} and \form{\FV(\pure'^{df}_{12})}
are observable variables.

As \form{\D_{comp}} is still open (it is not unsatisfiable) and \form{\pure'^{df}_{11}} is in a decidable fragment,
there exists \form{\sstack} such that
\form{\sstack {\force} \pure'^{df}_{11}}.
This implies that there exists \form{\sstack_1}
such that \form{\sstack {\subseteq} \sstate_1},
 \form{\sstack_1 {\force} \exists \setvars{w}_1{\cdot} \pure'^{df}_{11}}.
This means \form{\true {\implies} {\exists} \setvars{w}_1{\cdot} \pure'^{df}_{11}} (1).

Note if \form{\D_{comp}} is a descendant of \form{\D_{bud}},
\form{\pure'^{df}_{22}} must be of the form \form{\pure'^{df}_{22}{\equiv}\pure^{df}_{12}{\wedge}\pure^{df}_{12}{\wedge}...{\wedge}\pure^{df}_{12}}.
Thus, \form{\pure'^{df}_{22} \implies \pure^{df}_{12}} (2).

From (1) and (2), we have \form{\pure'^{df}_2 \implies \pure^{df}_1}.
\end{proof}

\noindent Secondly, we prove that
{\algonameesl} terminates for formulae including spatial-based
inductive predicates, those predicates
whose parameters are pointer-typed.
The quantifier elimination of function \form{\satb{\base}}  implies
that the satisfiability problem in this fragment
does not rely on existentially quantified 
 variables.
In other words, this problem
is based on free 
 variables.
As the number of these free variables in a formula
 is finite,
checking satisfiability in this fragment
 is decidable.
%
\begin{lemma}\label{lem.shape}
{\algonameesl} terminates for a system of
spatial-based
inductive predicates. 
\end{lemma}

\begin{proof}
In the following,
we use \form{\setvars{v}} to denote a sequence of variables.
Furthermore, we use \form{v_i} to denote
the \sm{i^{th}} variable in the sequence \form{\setvars{v}}.
Two sequence \form{\setvars{v}} and \form{\setvars{t}}
have the same sequence of observable variables
if for all valid position \form{k} then \form{\setvars{v_k}}
is an observable variable or \form{\nil} and \form{\setvars{t_k}{=}\setvars{v_k}}.

Suppose we are constructing a proof for the input
 \form{\D_0} where 
 \form{\D_0{\equiv}\exists \setvars{w}{\cdot}\base{\sep}
{\code{P_1}(\setvars{v}_{1}){\sep}..{\sep}{\code{P_{N}}(\setvars{v}_{N})}}}
and \form{\D_0} has \form{m} observable variables.
We are computing the longest path from \form{\D_0}
to a leaf \form{\D_{bud}} such that \form{\D_{bud}} can be linked back
to \form{\D_0}. In the worst case, this back-link
is established if (i) all points-to predicates are linked,
(ii) inductive predicates are linked and (iii) (dis)equalities constraints
between them are identical.
The complexity result is computed based on the following three facts.
\begin{enumerate}
\item There are $\mathcal{O}(2^m)$ possibilities of the spatial conjunction
of points-to predicates.
\item There are $\mathcal{O}(2^m)$ possibilities of sequence
of parameters of each inductive. And we have $\mathcal{O}(N)$ inductive
predicates.
\item The number of equality conjunctions over \form{m}
observable variables (including \form{\nil})
is \form{2m^2}. Thus, there are $\mathcal{O}(2^{2m^2})$ possibilities
of equalities.
\end{enumerate}
Thus, the longest path of an unfolding tree derived for \form{\D_0}
 is:
$\mathcal{O}(2^{m}\times 2^{2m^2} \times (2^{m})^N)$
\end{proof}




 Lastly, we show 
that {\algonameesl} always terminates for formulas
which contain inductive predicates with shape, pure constraints over data fields and
periodic relations over arithmetical parameters. We show that given a formula in {\slap},
{\algonameesl} can always construct a cyclic proof in finite time.


The main Theorem is shown based on the following.
\begin{enumerate}
 \item The termination does not rely on those constraints
on data fields of points-to predicates (lemma \ref{lem.arith}).
\item {\algonameesl} runs in $\mathcal{O}(2^{m}\times 2^{2m^2} \times (2^{m})^N)$
to link shape part of a leaf to one of its descendant (lemma \ref{lem.shape}).
 \item We now show the implication checking on pure formulas always hold
after shape parts have been linked.
For simplicity, we consider one arithmetical parameter \form{n}
of one inductive definition with periodic relation \form{R}.
Let \form{\pure_n} be the pure formula of the input relevant to
\form{n}, \form{S} be the periodic set
of \form{R} and \form{inv_n} be its corresponding Presburger formula.
There are two following subcases:
\begin{itemize}
\item \form{\pure_n {\wedge}\inv_n \implies \false}. This cace is detected by over-approximation step.
\item Otherwise, \form{\pure_n {\wedge} \inv_n} is satisfied.
We notice that \form{\pure_n {\wedge} \inv_n} must imply \form{\inv_n}.
In a path, the arithmetical part of a {\slap} bud is a disjunct of
the unfolding from its descendant.
The relation \form{R} is well-founded relation.
If  \form{R} is an increasing relation,
applying \form{R} over  \form{\pure_n {\wedge} \inv_n}
will be bounded by \form{\inv_n}. Thus, {\algonameesl} can link the pure parts
in finite time.
If \form{R} is a decreasing relation,
the periodic set \form{S'} derived for this disjunct
is a subset of the set derived for its descendant.
Thus, the implication at the last step
of function {\lbsl} always holds.
That is, {\lbsl} can always link the arithmetical part of
such above leaf to its descendant nodes.
\end{itemize}
Hence, {\algonameesl} always terminates.
 \hfill 
\end{enumerate}


 \end{subappendices}


\end{document}